\documentclass[12pt,a4paper,english,american]{article}
\usepackage{amsmath}
\usepackage{amssymb}

\makeatletter
\newcommand{\lyxaddress}[1]{
\par {\raggedright #1
\vspace{1.4em}
\noindent\par}
}

\usepackage{amsthm}
\usepackage{mathrsfs}

\newtheorem{thm}{Theorem}
\newtheorem{prop}[thm]{Proposition}
\newtheorem{lem}[thm]{Lemma}

\theoremstyle{remark}
\newtheorem*{rem*}{Remark}
\newtheorem*{rems*}{Remarks}

\theoremstyle{definition}
\newtheorem{rem}[thm]{Remark}
\newtheorem{define}[thm]{Definition}
\newtheorem*{def*}{Definition}

\newcommand{\sH}{\mathscr{H}}
\newcommand{\sL}{\mathscr{L}}
\newcommand{\sI}{\mathscr{I}}
\newcommand{\sF}{\mathscr{F}}

\newcommand{\sD}{\mathscr{D}}

\renewcommand{\Re}{\mathop\mathrm{Re}\nolimits}

\newcommand{\Dom}{\mathop\mathrm{Dom}\nolimits}
\newcommand{\supp}{\mathop\mathrm{supp}\nolimits}
\newcommand{\Tr}{\mathop\mathrm{Tr}\nolimits}

\newcommand{\dist}{\mathop\mathrm{dist}\nolimits}
\newcommand{\End}{\mathop\mathrm{End}\nolimits}
\newcommand{\Lin}{\mathop\mathrm{Lin}\nolimits}
\newcommand{\spec}{\mathop\mathrm{spec}\nolimits}

\makeatother

\usepackage{babel}

\begin{document}

\title{Noncommutative Bloch analysis of Bochner Laplacians with nonvanishing
gauge fields}

\author{P.~Ko\v{s}\v{t}\'akov\'a, P.~\v{S}\v{t}ov\'\i\v{c}ek}

\date{{}}

\maketitle

\lyxaddress{Department of Mathematics, Faculty of Nuclear Science,
  Czech Technical University in Prague, Trojanova13, 12000 Praha,
  Czech Republic}
\begin{abstract}
  \noindent Given an invariant gauge potential and a periodic scalar
  potential $\tilde{V}$ on a Riemannian manifold $\tilde{M}$ with a
  discrete symmetry group $\Gamma$, consider a $\Gamma$-periodic
  quantum Hamiltonian
  $\tilde{H}=-\tilde{\Delta}_{\mathrm{B}}+\tilde{V}$ where
  $\tilde{\Delta}_{\mathrm{B}}$ is the Bochner Laplacian. Both the
  gauge group and the symmetry group $\Gamma$ can be noncommutative,
  and the gauge field need not vanish. On the other hand, $\Gamma$ is
  supposed to be of type~I. With any unitary representation $\Lambda$
  of $\Gamma$ one associates a Hamiltonian
  $H^{\Lambda}=-\Delta_{\mathrm{B}}^{\Lambda}+V$ on
  $M=\tilde{M}/\Gamma$ where $V$ is the projection of $\tilde{V}$ to
  $M$. We describe a construction of the Bloch decomposition of
  $\tilde{H}$ into a direct integral whose components are
  $H^{\Lambda}$, with $\Lambda$ running over the dual space
  $\hat{\Gamma}$. The evolution operator and the resolvent decompose
  correspondingly. Conversely, given $\Lambda\in\hat{\Gamma}$, one can
  express the propagator $\mathcal{K}_{t}^{\Lambda}(y_{1},y_{2})$ (the
  kernel of $\exp(-itH^{\Lambda})$) in terms of the propagator
  $\tilde{\mathcal{K}}_{t}(y_{1},y_{2})$ (the kernel of
  $\exp(-it\tilde{H})$) as a weighted sum over $\Gamma$. Such a
  formula is known in theoretical physics for the case when the gauge
  field vanishes and $\tilde{M}$ is a universal covering space of a
  multiply connected manifold $M$.  We show that these constructions
  are mutually inverse. Analogous formulas exist for resolvents and
  their kernels (Green functions) as well.
\end{abstract}
\vskip\baselineskip\noindent \emph{Keywords}: Bloch analysis, periodic
Hamiltonian, covering space, propagator, Green function

\section{Introduction}

Suppose there is given a connected Riemannian manifold $\tilde{M}$
with a discrete symmetry group $\Gamma$, and an invariant Hermitian
vector fiber bundle with connection
$(\tilde{\mathfrak{V}},\tilde{\mathfrak{h}},\tilde{\nabla})$ over
$\tilde{M}$. This means that
$(\tilde{\mathfrak{V}},\tilde{\mathfrak{h}},\tilde{\nabla})%
=\pi^{\ast}(\mathfrak{V},\mathfrak{h},\nabla)$ where
$(\mathfrak{V},\mathfrak{h},\nabla)$ is a Hermitian vector fiber
bundle with connection over $M=\tilde{M}/\Gamma$ and
$\pi:\tilde{M}\to{}M$ is the projection. Let us consider a
$\Gamma$-periodic Hamilton operator in $L^{2}(\tilde{\mathfrak{V}})$
of the form $\tilde{H}=-\tilde{\Delta}_{\mathrm{B}}+\tilde{V}$ where
$\tilde{\Delta}_{\mathrm{B}}$ is the Bochner Laplacian and $\tilde{V}$
is a $\Gamma$-invariant semibounded real function on $\tilde{M}$
($\tilde{V}$ is the pull-back of a function $V$ on $M$). For any
finite-dimensional unitary representation $\Lambda$ of $\Gamma$ in a
vector space $\sL_{\Lambda}$ one constructs a Hermitian vector fiber
bundle $\mathfrak{V}^{\Lambda}$ over $M$ with a connection
$\nabla^{\Lambda}$, and consequently a Hamiltonian
$H^{\Lambda}=-\Delta_{\mathrm{B}}^{\Lambda}+V$ in
$L^{2}(\mathfrak{V}^{\Lambda})$.  An important feature of the
construction is that the operator $\tilde{H}$ decomposes into a direct
integral with components $H^{\Lambda}$ where $\Lambda$ runs over all
equivalence classes of irreducible unitary representations of
$\Gamma$. The evolution operator and the resolvent decompose
correspondingly.

This type of construction is well known in the cases when either the
connection $\nabla$ is flat (and so
$\Delta_{\mathrm{B}}=\Delta_{\mathrm{LB}}$ is the Laplace-Beltrami
operator) \cite{sunada,adachisunada,adachisunadasy,lledopost}, or the
group $\Gamma$ is commutative \cite{aschetal,gruber1}. The general
case when the connection is not flat and the symmetry group $\Gamma$
need not be commutative is treated in \cite{gruber2} in the framework
of noncommutative geometry using the theory of $C^{\ast}$ algebras.
The subject of the current paper is, too, an extension of the
construction of the Bloch decomposition to such a general case.  In
contrast to \cite{gruber2}, however, our construction relies on
standard techniques of differential geometry and, in particular, the
theory of $C^{\ast}$ algebras is not employed at all. We follow
closely the presentation given in \cite{kocabovastovicek} for the
particular case of a flat connection and a noncommutative symmetry
group $\Gamma$.

The action of $\Gamma$ on $\tilde{M}$ is usually supposed to be
co-compact
\cite{sunada,adachisunada,adachisunadasy,bruningsunsada,aschetal,gruber1,%
  lledopost}. Following \cite{kocabovastovicek} we relax this
assumption while introducing the Hamiltonian
$-\tilde{\Delta}_{\mathrm{B}}+\tilde{V}$ as the Friedrichs extension
of the corresponding symmetric operator defined on smooth sections
with compact supports.

In the framework of Feynman path integrals there was derived a
remarkable formula relating the propagators
$\mathcal{K}_{t}^{\Lambda}(x,x_{0})$ and
$\tilde{\mathcal{K}}_{t}(x,x_{0})$ associated respectively with the
Hamiltonians $H^{\Lambda}$ and $\tilde{H}$ \cite{schulman1,schulman2}.
A similar relation is known to hold for heat kernels \cite{atiyah}.
Another important application in mathematics of this type of formula
is in the derivation of the Selberg trace formula
\cite{selberg,hejhal}.  We show that such a formula makes sense also
in the more general case with a nonvanishing gauge field. Moreover, an
analogous formula can be derived for Green functions.

\section{Basic notions and notation \label{sec:periodic_H}}

All geometric objects are supposed to be smooth. Let $\tilde{M}$ be a
connected Riemannian manifold (Hausdorff and second countable) with a
discrete and at most countable symmetry group $\Gamma$. The action of
$\Gamma$ on $\tilde{M}$ is assumed to be smooth (every
$\gamma\in\Gamma$ acts as a diffeomorphism on $\tilde{M}$), free,
isometric and properly discontinuous. Denote by $\tilde{\mu}$ the
measure on $\tilde{M}$ induced by the Riemannian metric. The quotient
$M=\tilde{M}/\Gamma$ is a connected Riemannian manifold with an
induced measure $\mu$. The factorization defines a principal fiber
bundle $\pi:\tilde{M}\to M$ with the structure group $\Gamma$. All
$L^{2}$ spaces based on manifolds $M$ and $\tilde{M}$ are everywhere
understood with the measures $\mu$ and $\tilde{\mu}$, respectively.

To be specific, let us recall that the assumption on the properly
discontinuous action implies that for any compact set
$K\subset\tilde{M}$, the intersection $K\cap\gamma\cdot K$ is nonempty
for at most finitely many elements $\gamma\in\Gamma$. Moreover, any
point $y\in\tilde{M}$ has a neighborhood $U$ such that the sets
$\gamma\cdot U$, $\gamma\in\Gamma$, are mutually disjoint (see, for
instance, \cite{lee}).

Furthermore, assume that on $\tilde{M}$ there is given a
$\Gamma$-invariant gauge potential. Geometrically this means that over
$\tilde{M}$ there is given an invariant Hermitian vector fiber bundle
with connection (covariant derivative)
$(\tilde{\mathfrak{V}},\tilde{\mathfrak{h}},\tilde{\nabla})$. That is,
the action of $\Gamma$ on $\tilde{M}$ lifts to an action $\Psi$ on
$\tilde{\mathfrak{V}}$ which is fiber-wise linear, and both the
Hermitian product on fibers, $\tilde{\mathfrak{h}}$, and the covariant
derivative $\tilde{\nabla}$ are invariant with respect to the action
$\Psi$.

Let us denote by $L_{\gamma}$ the left action of $\gamma\in\Gamma$ on
$\tilde{M}$, i.e.\ $L_{\gamma}(y)=\gamma\cdot y$ for $y\in\tilde{M}$,
and by $C^{\infty}(\tilde{\mathfrak{V}})$ the vector space of smooth
sections of $\tilde{\mathfrak{V}}$. Notice that the invariance of
$\tilde{\nabla}$ means that for any $\gamma\in\Gamma$ fixed and at any
point $y\in\tilde{M}$,
\begin{equation}
  \forall\varphi\in{}C^{\infty}(\tilde{\mathfrak{V}}),
  \forall X\in{}T_{y}\tilde{M},\quad
  \tilde{\nabla}_{X}\left(\Psi_{\gamma^{-1}}
    \varphi(\gamma\cdot y)\right)
  =\Psi_{\gamma^{-1}}\!\left(\tilde{\nabla}_{\gamma\cdot X}\,
    \varphi(\gamma\cdot y)\right)
  \in\tilde{\mathfrak{V}}_{y},\label{eq:inv_nabla_tild}
\end{equation}
where $\gamma\cdot{}X\equiv(\mbox{d}L_{\gamma})_{y}X\in%
T_{\gamma\cdot{}y}\tilde{M}$. This property can be reformulated as
follows. Let $W_{\gamma}$, $\gamma\in\Gamma$, be the one-parameter
family of linear operators on $C^{\infty}(\tilde{\mathfrak{V}})$
defined by
\begin{equation}
  \forall\varphi\in C^{\infty}(\tilde{\mathfrak{V}}),\quad
  (W_{\gamma}\varphi)(y)=\Psi_{\gamma}\varphi(\gamma^{-1}\cdot y).
  \label{eq:def_Us}
\end{equation}
One clearly has $W_{\gamma}W_{\gamma'}=W_{\gamma\gamma'}$,
$\forall\gamma,\gamma'\in\Gamma$. Relation (\ref{eq:inv_nabla_tild})
means that for all smooth vector fields $\xi\in
C^{\infty}(T\tilde{M})$,
\[
\forall\gamma\in\Gamma,\;\tilde{\nabla}_{\xi}W_{\gamma}
=W_{\gamma}\tilde{\nabla}_{\gamma^{-1}\cdot\xi}\,.
\]

The assumptions mean that
$(\tilde{\mathfrak{V}},\tilde{\mathfrak{h}},\tilde{\nabla})$ admits a
factorization with respect to the action of $\Gamma$, and thus over
$M$ there exists a Hermitian vector fiber bundle with connection,
$(\mathfrak{V},\mathfrak{h},\nabla)$, such that
$(\tilde{\mathfrak{V}},\tilde{\mathfrak{h}},\tilde{\nabla})%
=\pi^{\ast}(\mathfrak{V},\mathfrak{h},\nabla)$ (the usual pull-back by
the projection $\pi:\tilde{M}\to M$). Conversely, any such a pull-back
is naturally $\Gamma$-invariant. The vector space
$C^{\infty}(\mathfrak{V})$ is identified with the subspace in
$C^{\infty}(\tilde{\mathfrak{V}})$ formed by those smooth sections
$\varphi$ which satisfy $W_{\gamma}\varphi=\varphi$,
$\forall\gamma\in\Gamma$; in more detail,
\begin{equation}
  \forall\gamma\in\Gamma,
  \forall y\in\tilde{M},\;\varphi(\gamma\cdot y)
  =\Psi_{\gamma}\varphi(y).
  \label{eq:sec_equivar}
\end{equation}
Let $\xi\in C^{\infty}(TM)$ be a vector field on $M$ and
$\tilde{\xi}\in C^{\infty}(T\tilde{M})$ be the unique vector field on
$\tilde{M}$ such that $\textrm{d}\pi(\tilde{\xi})=\xi$ (hence
$\tilde{\xi}$ is $\Gamma$-invariant). If $\varphi\in
C^{\infty}(\tilde{\mathfrak{V}})$ fulfills (\ref{eq:sec_equivar}) then
$\tilde{\nabla}_{\tilde{\xi}}\varphi$ fulfills (\ref{eq:sec_equivar})
as well. This defines the covariant derivative $\nabla_{\xi}$ in
$\mathfrak{V}$.

The Bochner Laplacian $\tilde{\Delta}_{\mathrm{B}}$ is a second-order
differential operator acting on smooth sections of $\tilde{\mathfrak{V}}$
whose construction depends on the covariant derivative $\tilde{\nabla}$
and on the Riemannian metric $\tilde{\mathfrak{g}}$ defined on cotangent
spaces on $\tilde{M}$. If $\varphi\in C^{\infty}(\tilde{\mathfrak{V}})$
then $\tilde{\nabla}\varphi$ belongs to $C^{\infty}(T^{\ast}\tilde{M}\otimes\tilde{\mathfrak{V}})$.
The differential operator $\tilde{\Delta}_{\mathrm{B}}$ is unambiguously
determined by the equality \begin{equation}
\forall\varphi_{1},\varphi_{2}\in C_{0}^{\infty}(\tilde{\mathfrak{V}}),\quad\int_{\tilde{M}}\tilde{\mathfrak{h}}(\varphi_{1},-\tilde{\Delta}_{\mathrm{B}}\varphi_{2})\,\mathrm{d}\tilde{\mu}=\int_{\tilde{M}}\tilde{\mathfrak{g}}\otimes\tilde{\mathfrak{h}}(\tilde{\nabla}\varphi_{1},\tilde{\nabla}\varphi_{2})\,\mathrm{d}\tilde{\mu}.\label{eq:BochnerLtilde}\end{equation}
The Bochner Laplacian $\tilde{\Delta}_{\mathrm{B}}$ is $\Gamma$-invariant
in the sense that, using the defining relation (\ref{eq:def_Us}),
\[
\forall\gamma\in\Gamma,\;\tilde{\Delta}_{\mathrm{B}}W_{\gamma}
=W_{\gamma}\tilde{\Delta}_{\mathrm{B}}.
\]

Analogously, one introduces the Bochner Laplacian $\Delta_{\mathrm{B}}$
on $M$ which is associated with the covariant derivative $\nabla$
and with the Riemannian metric $\mathfrak{g}$. If $\sigma\in C^{\infty}(\mathfrak{V})$
is represented by $\varphi\in C^{\infty}(\tilde{\mathfrak{V}})$ fulfilling
(\ref{eq:sec_equivar}) then $\Delta_{\mathrm{B}}\sigma$ is represented
by $\tilde{\Delta}_{\mathrm{B}}\varphi$.

Finally we summarize several basic facts concerning harmonic analysis
on $\Gamma$. In general, the harmonic analysis is well established
for locally compact groups of type~I \cite{shtern}, and this is
why we assume in the sequel that $\Gamma$ belongs to this class.
In addition, all irreducible representations of type~I groups are
finite-dimensional, and even the dimension is uniformly bounded \cite[Korollar~I]{thoma}.
This fact facilitates various algebraic constructions throughout the
paper. On the other hand, it is known that a countable discrete group
is of type~I if and only if it has an Abelian normal subgroup of
finite index \cite[Satz~6]{thoma}. This means, unfortunately, that
there exist covering spaces of interest whose structure group $\Gamma$
is not of type~I, and here we do not treat such cases.

Let $\hat{\Gamma}$ be the dual space to $\Gamma$ (the quotient space
of the space of irreducible unitary representations of $\Gamma$). In
the case in question the Haar measure on $\Gamma$ is nothing but the
counting measure. Let $\hat{m}$ be the Plancherel measure on
$\hat{\Gamma}$. For $\Lambda\in\hat{\Gamma}$ denote by
$\sI(\sL_{\Lambda})$ the space of linear maps on $\sI_{\Lambda}$.
Note that one can naturally identify
$\sI(\sL_{\Lambda})\equiv\sL_{\Lambda}\otimes\sL_{\Lambda}^{\ast}$
($\sL_{\Lambda}^{\ast}$ is the dual space to $\sL_{\Lambda}$) and thus
$\sI(\sL_{\Lambda})$ becomes equipped with the scalar product
$\langle{}A_{1},A_{2}\rangle=\Tr(A_{1}^{\,\ast}A_{2})$,
$\forall{}A_{1},A_{2}\in\sI(\sL_{\Lambda})$.

The Fourier transform is defined as a unitary mapping \[
\sF:L^{2}(\Gamma)\to\int_{\hat{\Gamma}}^{\oplus}\sI(\sL_{\Lambda})\,\mathrm{d}\hat{m}(\Lambda).\]
For $f\in L^{1}(\Gamma)\subset L^{2}(\Gamma)$ one has \[
\sF[f](\Lambda)=\sum_{\gamma\in\Gamma}f(\gamma)\Lambda(\gamma).\]
Conversely, if $f$ is of the form $f=g\ast h$ (the convolution)
where $g,h\in L^{1}(\Gamma)$, and $\hat{f}=\sF[f]$ then \[
f(s)=\int_{\hat{\Gamma}}\Tr[\Lambda(s)^{\ast}\hat{f}(\Lambda)]\,\mathrm{d}\hat{m}(\Lambda).\]

Using the fact that $\Gamma$ is a countable discrete group as well
as the unitarity of the Fourier transform one finds that \[
\hat{m}(\hat{\Gamma})\leq\int_{\hat{\Gamma}}\dim\sL_{\Lambda}\,\mathrm{d}\hat{m}(\Lambda)=1.\]
The following rules satisfied by the Fourier transformation are also
of importance: \begin{equation}
\forall r\in\Gamma,\forall f\in L^{2}(\Gamma),\,\,\sF[f(r\cdot\gamma)](\Lambda)=\Lambda(r^{-1})\sF[f(\gamma)](\Lambda)\label{eq:Fourier_rule}\end{equation}
(here $\sF$ acts in the variable $\gamma\in\Gamma$), and, conversely,\begin{equation}
\forall r\in\Gamma,\forall\hat{f}\in\int_{\hat{\Gamma}}^{\oplus}\sI(\sL_{\Lambda})\,\mathrm{d}\hat{m}(\Lambda),\;\sF^{-1}[\Lambda(r)\hat{f}(\Lambda)](\gamma)=\sF^{-1}[\hat{f}(\Lambda)](r^{-1}\gamma)\label{eq:Fourier_rule_inv}\end{equation}
(here $\sF^{-1}$ acts in the variable $\Lambda\in\hat{\Gamma}$).

\section{A construction of the noncommutative Bloch decomposition \label{sec:bloch}}

\subsection{Associated vector fiber bundles over $M$ \label{sec:assoc_vect_bundles}}

Since $\dim\tilde{M}=\dim M$, in the principal fiber bundle $\tilde{M}$
over $M$ with the structure group $\Gamma$ there exists a unique
connection which is necessarily flat. Given a finite-dimensional unitary
representation $\Lambda$ of $\Gamma$ in $\sL_{\Lambda}$ one can
associate with the principal fiber bundle a vector fiber bundle $E(\Lambda)$
over $M$ with a typical fiber $\sL_{\Lambda}$ \cite{kobayashinomizu}.
Provided the representation is unitary $E(\Lambda)$ naturally acquires
a Hermitian structure. The flat connection in $\tilde{M}$ carries
over to the vector fiber bundle $E(\Lambda)$ as a Hermitian covariant
derivative which is again flat.

Suppose $(\mathfrak{V}_{j},\mathfrak{h_{j}},\nabla_{j})$, $j=1,2$,
are two Hermitian vector fiber bundles with connection over $M$.
The tensor product $\mathfrak{V}_{1}\otimes\mathfrak{V}_{2}$ is a
vector fiber bundle over $M$ with fibers $(\mathfrak{V}_{1}\otimes\mathfrak{V}_{2})_{x}=(\mathfrak{V}_{1})_{x}\otimes(\mathfrak{V}_{2})_{x}$,
$x\in M$, and it is again equipped with a Hermitian structure in
a canonical way. Moreover, a Hermitian connection $\nabla_{12}$ is
naturally defined in $\mathfrak{V}_{1}\otimes\mathfrak{V}_{2}$ by
the rule: for any vector field $\xi\in C^{\infty}(TM)$,\[
\forall\varphi_{1}\in C^{\infty}(\mathfrak{V}_{1}),\varphi_{2}\in C^{\infty}(\mathfrak{V}_{2}),\ \nabla_{12}(\xi)\varphi_{1}\otimes\varphi_{2}=(\nabla_{1}(\xi)\varphi_{1})\otimes\varphi_{2}+\varphi_{1}\otimes(\nabla_{2}(\xi)\varphi_{2}).\]
\begin{define} Let $(\mathfrak{V}^{\Lambda},\mathfrak{h}^{\Lambda},\nabla^{\Lambda})$
be the Hermitian vector fiber bundle with connection over $M$ obtained
as the tensor product of $(\mathfrak{V},\mathfrak{h},\nabla)$ and
the associated vector fiber bundle $E(\Lambda)$ (which is supposed
to be equipped with the Hermitian structure and the Hermitian covariant
derivative, as recalled above). \end{define}

Though the construction of associated fiber bundles is standard let
us indicate some intermediate objects occurring in the construction
for the sake of future reference. Denote by $(\tilde{\mathfrak{V}}^{\Lambda},\tilde{\mathfrak{h}}^{\Lambda},\tilde{\nabla}^{\Lambda})$
the Hermitian vector fiber bundle with connection over $\tilde{M}$
with fibers $\tilde{\mathfrak{V}}_{y}^{\Lambda}=\tilde{\mathfrak{V}}_{y}\otimes\sL_{\Lambda}$,
$y\in\tilde{M}$. The Hermitian product on $\tilde{\mathfrak{V}}_{y}^{\Lambda}$
is defined in the usual way. The covariant derivative $\tilde{\nabla}^{\Lambda}$
is defined so that for all $\varphi\in C^{\infty}(\tilde{\mathfrak{V}})$
and $v\in\sL_{\Lambda}$ one has\[
\tilde{\nabla}^{\Lambda}\varphi\otimes v=(\tilde{\nabla}\varphi)\otimes v\in C^{\infty}(T^{\ast}\tilde{M}\otimes\tilde{\mathfrak{V}}\otimes\sL_{\Lambda}).\]

Let $\Psi^{\Lambda}$ be the action of $\Gamma$ on
$\tilde{\mathfrak{V}}^{\Lambda}=\tilde{\mathfrak{V}}\otimes\sL_{\Lambda}$
defined by
\[
\Psi_{\gamma}^{\Lambda}=\Psi_{\gamma}\otimes\Lambda(\gamma),
\textrm{~}\forall\gamma\in\Gamma.
\]
Moreover, analogously to (\ref{eq:def_Us}) one introduces a
one-parameter family of linear operators on
$C^{\infty}(\tilde{\mathfrak{V}}^{\Lambda})$ called
$W_{\gamma}^{\Lambda}$, $\gamma\in\Gamma$, i.e.\ one puts
\[
W_{\gamma}^{\Lambda}=W_{\gamma}\otimes\Lambda(\gamma),\;\gamma\in\Gamma.
\]
Observe that, with this definition,
$(\tilde{\mathfrak{V}}^{\Lambda},\tilde{\mathfrak{h}}^{\Lambda},%
\tilde{\nabla}^{\Lambda})$ is again $\Gamma$-invariant. Furthermore,
very similarly to (\ref{eq:BochnerLtilde}), one introduces the Bochner
Laplacian $\tilde{\Delta}_{\mathrm{B}}^{\Lambda}$ as a differential
operator acting on smooth sections of
$\tilde{\mathfrak{V}}^{\Lambda}$, and one readily finds that
\[
\forall\varphi\in C^{\infty}(\tilde{\mathfrak{V}}),
\forall v\in\sL_{\Lambda},
\quad\tilde{\Delta}_{\mathrm{B}}^{\Lambda}\,\varphi\otimes v
=(\tilde{\Delta}_{\mathrm{B}}\varphi)\otimes v.
\]
Note that the Bochner Laplacian
$\tilde{\Delta}_{\mathrm{B}}^{\Lambda}$ commutes with all
$W_{\gamma}^{\Lambda}$, $\gamma\in\Gamma$.

$(\mathfrak{V}^{\Lambda},\mathfrak{h}^{\Lambda},\nabla^{\Lambda})$
is in fact nothing but the factorization of $(\tilde{\mathfrak{V}}^{\Lambda},\tilde{\mathfrak{h}}^{\Lambda},\tilde{\nabla}^{\Lambda})$
with respect to the action $\Psi^{\Lambda}$ of $\Gamma$. Again one
has $(\tilde{\mathfrak{V}}^{\Lambda},\tilde{\mathfrak{h}}^{\Lambda},\tilde{\nabla}^{\Lambda})=\pi^{\ast}(\mathfrak{V}^{\Lambda},\mathfrak{h}^{\Lambda},\nabla^{\Lambda})$.
It is convenient to identify smooth (or measurable) sections $\psi$
of $\mathfrak{V}^{\Lambda}$ with smooth (measurable) sections $\varphi$
of $\tilde{\mathfrak{V}}^{\Lambda}$ fulfilling\begin{equation}
\forall\gamma\in\Gamma,\:\varphi(\gamma\cdot y)=\Psi_{\gamma}^{\Lambda}\varphi(y)\label{eq:sec_Lambda_equiv}\end{equation}
everywhere (or almost everywhere) on $\tilde{M}$. In that case we
say that $\varphi$ is an equivariant section.

The Bochner Laplacian $\Delta_{\mathrm{B}}^{\Lambda}$ associated
with $\nabla^{\Lambda}$ and $\mathfrak{g}$ is introduced on $M$
in the standard way, similarly to (\ref{eq:BochnerLtilde}).

We conclude this subsection with an auxiliary construction.

\begin{define}
  Let us define
  $\Phi^{\Lambda}:C_{0}^{\infty}(\tilde{\mathfrak{V}}^{\Lambda})\to%
  C_{0}^{\infty}(\mathfrak{V}^{\Lambda})$ so that
  $\forall\sigma\in{}C_{0}^{\infty}(\tilde{\mathfrak{V}}^{\Lambda})$,
  $\Phi^{\Lambda}\sigma$ is represented by the series
  \begin{equation}
    \varphi=\sum_{\gamma\in\Gamma}W_{\gamma}^{\Lambda}\sigma
    \in C_{0}^{\infty}(\tilde{\mathfrak{V}}^{\Lambda}),\ \mbox{i.e.\ \ }
    \varphi(y)=\sum_{\gamma\in\Gamma}\Psi_{\gamma}^{\Lambda}
    \sigma(\gamma^{-1}\cdot y).
    \label{eq:PhiLamda_def}
  \end{equation}
\end{define}

\begin{rem*}
  Note that $\varphi$ is in fact a smooth section of
  $\tilde{\mathfrak{V}}^{\Lambda}$ for the action of $\Gamma$ is
  properly discontinuous. Moreover, $\varphi$ fulfills
  (\ref{eq:sec_Lambda_equiv}) and thus it represents a smooth section
  of $\mathfrak{V}^{\Lambda}$ whose support is contained in
  $\pi(\supp\sigma)$ and so is compact.
\end{rem*}

\begin{lem}
  \label{thm:PhiL_surject}
  The range of $\Phi^{\Lambda}$ is equal to the whole space
  $C_{0}^{\infty}(\mathfrak{V}^{\Lambda})$.
\end{lem}

\begin{proof}
  Given $\psi\in C_{0}^{\infty}(\mathfrak{V}^{\Lambda})$ one can
  assume, without loss of generality, that there exists an open
  neighborhood $U\supset\supp\psi$ such that the principal fiber
  bundle $\pi:\tilde{M}\to M$ is trivial over $U$. Let
  $\eta:U\to\eta(U)\subset\tilde{M}$ be a smooth section, and
  $\varphi\in C^{\infty}(\tilde{\mathfrak{V}}^{\Lambda})$ be the
  equivariant section (i.e.\ fulfilling (\ref{eq:sec_Lambda_equiv}))
  representing $\psi$. Then the restriction $\varphi|_{\eta(U)}$ is a
  smooth section of $\tilde{\mathfrak{V}}^{\Lambda}$ over the open set
  $\eta(U)$ with a compact support, and it extends naturally to a
  global section
  $\sigma\in{}C_{0}^{\infty}(\tilde{\mathfrak{V}}^{\Lambda})$
  (vanishing outside of $\eta(U)$). Both $\varphi$ and the series
  $\sum_{\gamma\in\Gamma}W_{\gamma}^{\Lambda}\sigma$ fulfill
  (\ref{eq:sec_Lambda_equiv}).  Moreover, the two sections coincide on
  $\eta(U)$ and so they coincide everywhere. Hence
  $\Phi^{\Lambda}\sigma=\psi$.
 \end{proof}

\begin{rem} Using once more the fact that the action of $\Gamma$
is properly discontinuous and that $\tilde{\Delta}_{\mathrm{B}}^{\Lambda}$
commutes with all operators $W_{s}^{\Lambda}$ one observes that \begin{equation}
\forall\sigma\in C_{0}^{\infty}(\tilde{\mathfrak{V}}^{\Lambda}),\;\Delta_{\mathrm{B}}^{\Lambda}\,\Phi^{\Lambda}\,\sigma=\Phi^{\Lambda}\tilde{\Delta}_{\mathrm{B}}^{\Lambda}\,\sigma.\label{eq:LaplacePhi_comm}\end{equation}
\end{rem}

\subsection{The Bloch decomposition}

We remind the reader that $\Gamma$ is assumed to be of type~I. As
already recalled in Section~\ref{sec:periodic_H}, this in particular
means that all irreducible representations of $\Gamma$ are finite-dimensional.
Thus all representation Hilbert spaces $\sL_{\Lambda}$ are finite-dimensional,
which is why all tensor products in the construction of associated
vector bundles, as described in Subsection~\ref{sec:assoc_vect_bundles},
make sense.

Suppose we are given a real measurable $\Gamma$-invariant function
$\tilde{V}$ on $\tilde{M}$ bounded from below. Hence $\tilde{V}=\pi^{\ast}V$
for a basically unique real measurable function $V$ on $M$ which
is bounded from below as well. The differential operator $-\tilde{\Delta}_{\mathrm{B}}+\tilde{V}$
is well defined on the domain $C_{0}^{\infty}(\tilde{\mathfrak{V}})\subset L^{2}(\tilde{\mathfrak{V}})$.
Moreover, since this is a densely defined operator which is readily
seen to be symmetric and semibounded, one can apply a standard procedure
resulting in a distinguished selfadjoint extension with the same lower
bound, the so called Friedrichs extension \cite[Chp.\ VI\ \S2]{kato}.
This extension is in some sense minimal (its form domain is the smallest
one among all selfadjoint extensions), and the basic steps of its
construction are roughly as follows. Using the differential operator
$-\tilde{\Delta}_{\mathrm{B}}+\tilde{V}$ one defines, in a standard
manner, a quadratic form on the domain $C_{0}^{\infty}(\tilde{\mathfrak{V}})$.
Since this densely defined quadratic form is semibounded, it is closable.
Now it is known (by a result which is sometimes called the first representation
theorem) that in turn a unique selfadjoint operator is associated
with the closed semibounded quadratic form, and this is exactly the
sought minimal selfadjoint extension (see also \cite[Chp.\ 5]{weidmann}
or \cite[Chp.\ X\ \S3]{reedsimon2}). Let us denote the Friedrichs
extension by $\tilde{H}$.

Furthermore, note that the linear operators $W_{\gamma}$, $\gamma\in\Gamma$,
defined in (\ref{eq:def_Us}) map bijectively the vector space $C_{0}^{\infty}(\tilde{\mathfrak{V}})$
onto itself and preserve the $L^{2}$ norm, and so they extend unambiguously
to unitary operators on $L^{2}(\tilde{\mathfrak{V}})$. The Hamiltonian
$\tilde{H}$ commutes with all unitary operators $W_{\gamma}$, and
in this sense it is $\Gamma$-periodic.

Let $\Lambda$ be again a finite-dimensional unitary representation
of $\Gamma$. Similarly as above, let us denote by $H^{\Lambda}$
the Friedrichs extension of the differential operator $-\Delta_{\mathrm{B}}^{\Lambda}+V$
defined on the domain $C_{0}^{\infty}(\mathfrak{V}^{\Lambda})\subset L^{2}(\mathfrak{V}^{\Lambda})$.
Note that if $\psi\in C_{0}^{\infty}(\mathfrak{V}^{\Lambda})$ is
represented by $\varphi\in C_{0}^{\infty}(\tilde{\mathfrak{V}}^{\Lambda})$
fulfilling (\ref{eq:sec_Lambda_equiv}) then $H^{\Lambda}\psi$ is
represented by $(-\tilde{\Delta}_{\mathrm{B}}^{\Lambda}+\tilde{V})\varphi$.

In the first step of the generalized Bloch analysis one decomposes
$\tilde{H}$ into a direct integral over $\hat{\Gamma}$ with components
being equal to $H^{\Lambda}$. As a corollary one obtains a similar
relationship for the corresponding evolution operators and resolvents.
The decomposition is achieved by applying a unitary mapping $\Phi$
described below. Its construction is basically a modification of the
construction presented in Section~IV of \cite{kocabovastovicek}, and
therefore we omit here some details and, first of all, some proofs
which resemble those given in \cite{kocabovastovicek}. In particular,
this is true for the proof of the following lemma.

\begin{lem}
  \label{thm:f_sub_y}
  For $f\in L^{2}(\tilde{\mathfrak{V}})$ and $y\in\tilde{M}$ put
  \begin{equation}
    \forall\gamma\in\Gamma,\; f_{y}(\gamma)
    =\Psi_{\gamma}f(\gamma^{-1}\cdot y)\in\tilde{\mathfrak{V}}_{y}.
    \label{eq:f_sub_y_def}
  \end{equation}
  Then $f_{y}$ is well defined for almost all $x\in M$ and all
  $y\in\pi^{-1}(\{x\})$ and belongs to
  $L^{2}(\Gamma,\tilde{\mathfrak{V}}_{y})\equiv\tilde{\mathfrak{V}}_{y}\otimes%
  L^{2}(\Gamma)$.
\end{lem}

Observe that the tensor product
$L^{2}(\mathfrak{V}^{\Lambda})\otimes\sL_{\Lambda}^{*}$ can be
identified with the Hilbert space formed by those measurable sections
$\psi$ of
$\tilde{\mathfrak{V}}^{\Lambda}\otimes\sL_{\Lambda}^{*}\equiv%
\tilde{\mathfrak{V}}\otimes\sI(\sL_{\Lambda})$ which satisfy
\begin{equation}
  \forall\gamma\in\Gamma,\;\psi(\gamma\cdot y)
  =\left(\Psi_{\gamma}\otimes L_{\Lambda(\gamma)}\right)\psi(y)
  \quad\textrm{a.e.\ on }\tilde{M},
  \label{eq:equiv_PsiLLambda}
\end{equation}
with $L_{\Lambda(\gamma)}\in\End(\sI(\sL_{\Lambda}))$ being the linear
operator on $\sI(\sL_{\Lambda})$ acting by multiplication from the
left, $L_{\Lambda(\gamma)}A=\Lambda(\gamma)A$,
$\forall{}A\in\sI(\sL_{\Lambda})$, and which have finite $L^{2}$ norms
(with integration taken over $M$). The next lemma readily follows from
the unitarity of the Fourier transformation and from property
(\ref{eq:Fourier_rule}).

\begin{lem}
  \label{thm:Fourier_on_fy}
  For any $f\in L^{2}(\tilde{\mathfrak{V}})$ let
  $f_{y}\in\tilde{\mathfrak{V}}_{y}\otimes L^{2}(\Gamma)$,
  $y\in\tilde{M}$, be as defined in (\ref{eq:f_sub_y_def}). Then the
  measurable section of the vector bundle
  $\tilde{\mathfrak{V}}\otimes\sI(\sL_{\Lambda})$ given by
  \begin{equation}
    \tilde{M}\ni y\mapsto(1\otimes\sF)[f_{y}](\Lambda)
    \in\tilde{\mathfrak{V}}_{y}\otimes\sI(\sL_{\Lambda})
    \equiv\left(\tilde{\mathfrak{V}}\otimes\sI(\sL_{\Lambda})\right)_{\! y},
    \label{eq:sec_defby_Fourier}
  \end{equation}
  is well defined for $\textrm{a.a.\ }\Lambda\in\hat{\Gamma}$, and for
  those $\Lambda$ it satisfies the equivariance condition
  (\ref{eq:equiv_PsiLLambda}) and so it represents a measurable
  section of $\mathfrak{V}^{\Lambda}\otimes\sL_{\Lambda}^{*}$. The
  section has a finite $L^{2}$ norm and thus it belongs to
  $L^{2}(\mathfrak{V}^{\Lambda})\otimes\sL_{\Lambda}^{*}$.
\end{lem}

\begin{define}
  We define 
  \begin{equation}
    \Phi:L^{2}(\tilde{\mathfrak{V}})
    \to\int_{\hat{\Gamma}}^{\oplus}L^{2}(\mathfrak{V}^{\Lambda})
    \otimes\sL_{\Lambda}^{\ast}\;\textrm{d}\hat{m}(\Lambda)
    \label{eq:Phi_intro}
  \end{equation}
  so that for all $f\in L^{2}(\tilde{\mathfrak{V}})$, the components
  $\Phi[f](\Lambda)$ are given by the prescription:
  \begin{equation}
    \mbox{for a.a.\ }\Lambda\in\hat{\Gamma},
    \mbox{\ a.a.\ }y\in\tilde{M},\ \Phi[f](\Lambda)\,(y)
    =(1\otimes\sF)[f_{y}](\Lambda)\in\left(\tilde{\mathfrak{V}}
      \otimes\sI(\sL_{\Lambda})\right)_{\! y}.
    \label{eq:Phi_def}
  \end{equation}
  According to Lemma~\ref{thm:Fourier_on_fy}, $\Phi[f](\Lambda)$
  fulfills (\ref{eq:equiv_PsiLLambda}) and can be identified with an
  element from
  $L^{2}(\mathfrak{V}^{\Lambda})\otimes\sL_{\Lambda}^{\ast}$. In
  particular, if
  $f\in{}L^{1}(\tilde{\mathfrak{V}})\cap{}L^{2}(\tilde{\mathfrak{V}})$
  then
  \begin{equation}
    \Phi[f](\Lambda)\,(y)
    =\sum_{\gamma\in\Gamma}\Psi_{\gamma}f(\gamma^{-1}\cdot y)
    \otimes\Lambda(\gamma).
    \label{eq:Phi_series}
  \end{equation}
\end{define}

\begin{lem}
  $\Phi$ is a unitary mapping.
\end{lem}

\begin{proof}
  Let $p_{2}:\Gamma\times\tilde{M}\to\tilde{M}$ be the projection onto
  the second component, and $\check{\sH}$ be the Hilbert space formed
  by measurable sections $\check{\psi}$ of
  $p_{2}^{\,\ast}\,\tilde{\mathfrak{V}}$ which satisfy
  \begin{equation}
    \forall r,\gamma\in\Gamma,\;
    \textrm{for a.a.\ }y\in\tilde{M},\;\check{\psi}(r,\gamma\cdot y)
    =\Psi_{\gamma}\check{\psi}(\gamma^{-1}r,y).
    \label{eq:psicheck_equiv}
  \end{equation}
  Hence
  $\check{\psi}(\gamma,y)=\Psi_{\gamma}\check{\psi}(1,\gamma^{-1}\cdot{}y)$.
  It follows that the function
  \[
  y\mapsto\sum_{\gamma\in\Gamma}\|\check{\psi}(\gamma,y)\|^{2}
  =\sum_{\gamma\in\Gamma}\|\check{\psi}(1,\gamma\cdot y)\|^{2}
  \]
  is $\Gamma$-invariant and projects to a function $\psi_{\ast}(x)^{2}$
  defined on $M$. The norm on $\check{\sH}$ is given by the integral
  \[
  \|\check{\psi}\|^{2}=\int_{M}\psi_{\ast}(x)^{2}\,\textrm{d}\mu(x).
  \]

  Consider the linear mapping
  \begin{equation}
    \Theta:L^{2}(\tilde{\mathfrak{V}})
    \to\check{\sH}:f\mapsto\check{\psi},
    \quad\check{\psi}(\gamma,y)=\Psi_{\gamma}f(\gamma^{-1}\cdot y).
    \label{eq:Theta_def}
  \end{equation}
  By a simple computation one can check that, for all $f\in
  L^{2}(\tilde{\mathfrak{V}})$, the image $\check{\psi}=\Theta f$
  actually fulfills (\ref{eq:psicheck_equiv}) and
  $\|\check{\psi}\|=\|f\|$. Moreover, $\Theta$ is clearly invertible,
  $(\Theta^{-1}\check{\psi})(y)=\check{\psi}(1,y)$. Hence $\Theta$ is
  a unitary mapping. Observe that if $\check{\psi}\in\check{\sH}$ then
  \[
  \mbox{for a.a.\ }y\in\tilde{M},\ \check{\psi}(\cdot,y)
  \in\tilde{\mathfrak{V}}_{y}\otimes L^{2}(\Gamma).
  \]

  Further let us set
  \[
  \hat{\sH}=\int_{\hat{\Gamma}}^{\oplus}L^{2}(\mathfrak{V}^{\Lambda})
  \otimes\sL_{\Lambda}^{\ast}\;\textrm{d}\hat{m}(\Lambda).
  \]
  Using an analogous identification as above, if
  $\hat{\psi}\in\hat{\sH}$ then
  $\hat{\psi}(\Lambda,y)\in\tilde{\mathfrak{V}}_{y}\otimes\sI(\sL_{\Lambda})$
  is defined almost everywhere on $\hat{\Gamma}\times\tilde{M}$ and
  fulfills
  \begin{equation}
    \forall\gamma\in\Gamma,\,\textrm{for a.a.\ }(\Lambda,y)
    \in\hat{\Gamma}\times\tilde{M},\;\hat{\psi}(\Lambda,\gamma\cdot y)
    =(\Psi_{\gamma}\otimes L_{\Lambda(\gamma)})\hat{\psi}(\Lambda,y).
    \label{eq:pshihat_equiv}
  \end{equation}
  Observe that if $\hat{\psi}\in\hat{\sH}$ then
  \[
  \mbox{for a.a.\ }y\in\tilde{M},\ \hat{\psi}(\cdot,y)
  \in\tilde{\mathfrak{V}}_{y}
  \otimes\int_{\hat{\Gamma}}^{\oplus}\sI(\sL_{\Lambda})\,
  \textrm{d}\hat{m}(\Lambda).
  \]

  Next we introduce two mutually inverse linear mappings,
  $\Xi:\check{\sH}\to\hat{\sH}:\check{\psi}\mapsto\hat{\psi}$, and
  $\Xi^{-1}:\hat{\sH}\to\check{\sH}:\hat{\psi}\mapsto\check{\psi}$,
  defined by the equalities
  \begin{eqnarray}
    \hat{\psi}(\Lambda,y) & = &
    (1\otimes\sF)[\check{\psi}(\cdot,y)](\Lambda)\,
    \in\,\tilde{\mathfrak{V}}_{y}\otimes\sI(\sL_{\Lambda}),
    \label{eq:Xi_def}\\
    \check{\psi}(\gamma,y) & = &
    (1\otimes\sF^{-1})[\hat{\psi}(\cdot,y)](\gamma)\,
    \in\,\tilde{\mathfrak{V}}_{y}.\nonumber
  \end{eqnarray}
  If $\check{\psi}$ satisfies (\ref{eq:psicheck_equiv}) then property
  (\ref{eq:Fourier_rule}) of the Fourier transformation implies that
  the image $\hat{\psi}=\Xi\check{\psi}$ fulfills
  \[
  \hat{\psi}(\Lambda,r\cdot y)
  =(1\otimes\sF)[\Psi_{r}\check{\psi}(r^{-1}\gamma,y)](\Lambda)
  =(\Psi_{r}\otimes L_{\Lambda(r)})\hat{\psi}(\Lambda,y).
  \]
  Conversely, if $\hat{\psi}$ satisfies (\ref{eq:pshihat_equiv}) then
  property (\ref{eq:Fourier_rule_inv}) of the Fourier transformation
  implies that the image $\check{\psi}=\Xi^{-1}\hat{\psi}$ fulfills
  \[
  \check{\psi}(r,\gamma\cdot y)
  =\Psi_{\gamma}(1\otimes\sF^{-1})[\left(1\otimes\Lambda(\gamma)\right)
  \hat{\psi}(\Lambda,y)](r)=\Psi_{\gamma}\check{\psi}(\gamma^{-1}r,y).
  \]
  The unitarity of the Fourier transformation implies that the
  mappings $\Xi$ and $\Xi^{-1}$ are unitary as well.

  Comparing the definitions of $\Phi$, $\Theta$ and $\Xi$ given in
  (\ref{eq:Phi_def}), (\ref{eq:Theta_def}) and (\ref{eq:Xi_def}),
  respectively, one can see that $\Phi=\Xi\Theta$. Hence $\Phi$ is
  unitary.
\end{proof}

\begin{rem} Alternatively, one can define $\Phi$ as follows. For
a given $\Lambda\in\hat{\Gamma}$, let us identify $L^{2}(\mathfrak{V}^{\Lambda})\otimes\sL_{\Lambda}^{\ast}$
with $\Lin(\sL_{\Lambda},L^{2}(\mathfrak{V}^{\Lambda}))$. For $\varphi\in C_{0}^{\infty}(\tilde{\mathfrak{V}})$,
define a linear mapping $\Phi[\varphi](\Lambda):\sL_{\Lambda}\to L^{2}(\mathfrak{V}^{\Lambda})$
by\begin{equation}
\forall v\in\sL_{\Lambda},\;\Phi[\varphi](\Lambda)v=\Phi^{\Lambda}\,\varphi\otimes v,\label{eq:Phi_PhiL_def}\end{equation}
with $\Phi^{\Lambda}$ being given in (\ref{eq:PhiLamda_def}). This
way one gets a linear mapping \[
\Phi:C_{0}^{\infty}(\tilde{\mathfrak{V}})\to\int_{\hat{\Gamma}}^{\oplus}L^{2}(\mathfrak{V}^{\Lambda})\otimes\sL_{\Lambda}^{\ast}\:\mathrm{d}\hat{m}(\Lambda)\]
that can be verified to be an isometry and so it unambiguously extends
from $C_{0}^{\infty}(\tilde{\mathfrak{V}})$ to $L^{2}(\tilde{\mathfrak{V}})$.
It is not difficult to see that definitions (\ref{eq:Phi_def}) and
(\ref{eq:Phi_PhiL_def}) of the mapping $\Phi$ in fact coincide.
\end{rem}

Now we are ready to describe the Bloch decomposition. Put
\[
\Sigma=\overline{\bigcup_{\Lambda\in\hat{\Gamma}}\spec(H^{\Lambda})}
\]
and, for $z$ running over the corresponding resolvent sets,
\[
\tilde{R}(z)=(\tilde{H}-z)^{-1},\mbox{ }R^{\Lambda}(z)
=(H^{\Lambda}-z)^{-1}.
\]
Furthermore,
\[
\tilde{U}(t)=\exp(-it\tilde{H}),\ U^{\Lambda}(t)
=\exp(-itH^{\Lambda}),\ t\in\mathbb{R}.
\]

\begin{thm}
  The unitary mapping $\Phi$ decomposes the Hamiltonian $\tilde{H}$,
  i.e.\
  \begin{equation}
    \Phi\tilde{H}\Phi^{-1}
    =\int_{\hat{\Gamma}}^{\oplus}H^{\Lambda}
    \otimes1\,\mathrm{d}\hat{m}(\Lambda).
    \label{eq:H_int_HL}
  \end{equation}
  Consequently,
  \begin{equation}
    \Phi\tilde{U}(t)\Phi^{-1}
    =\int_{\hat{\Gamma}}^{\oplus}U^{\Lambda}(t)
    \otimes1\,\mathrm{d}\hat{m}(\Lambda),\mbox{ }t\in\mathbb{R},
    \label{eq:U_int_UL}
  \end{equation}
  and
  \begin{equation}
    \Phi\tilde{R}(z)\Phi^{-1}
    =\int_{\hat{\Gamma}}^{\oplus}R^{\Lambda}(z)
    \otimes1\,\mathrm{d}\hat{m}(\Lambda),
    \mbox{ }z\in\mathbb{C}\setminus\Sigma.
    \label{eq:R_int_RL}
  \end{equation}
\end{thm}

\begin{proof} From (\ref{eq:Phi_PhiL_def}) one deduces that if $\varphi\in C_{0}^{\infty}(\tilde{\mathfrak{V}})$
then $\Phi[\varphi](\Lambda)\in C_{0}^{\infty}(\mathfrak{V}^{\Lambda})\otimes\sL_{\Lambda}^{\ast}$,
$\forall\Lambda\in\hat{\Gamma}$. Taking into account also (\ref{eq:LaplacePhi_comm})
one has\[
\left(\Delta_{\mathrm{B}}^{\Lambda}\otimes1\right)\Phi[\varphi](\Lambda)=\Phi[\tilde{\Delta}_{\mathrm{B}}\varphi](\Lambda).\]
Moreover, if $f\in L^{2}(\tilde{\mathfrak{V}})$ then $(\tilde{V}f)_{y}=\tilde{V}(y)f_{y}$
and so\[
\forall\Lambda\in\hat{\Gamma},\;\Phi[\tilde{V}f](\Lambda)=(V\otimes1)\Phi[f](\Lambda).\]
Altogether this implies that\[
\forall\varphi_{1},\varphi_{2}\in C_{0}^{\infty}(\tilde{\mathfrak{V}}),\forall\Lambda\in\hat{\Gamma},\;\Phi[(-\tilde{\Delta}_{\mathrm{B}}+\tilde{V})\varphi_{2}](\Lambda)=\left((-\Delta_{\mathrm{B}}^{\Lambda}+V)\otimes1\right)\Phi[\varphi_{2}](\Lambda),\]
and\begin{eqnarray*}
\langle\varphi_{1},(-\tilde{\Delta}_{\mathrm{B}}+\tilde{V})\varphi_{2}\rangle & = & \left\langle \Phi[\varphi_{1}],\Phi[(-\tilde{\Delta}_{\mathrm{B}}+\tilde{V})\varphi_{2}]\right\rangle \\
 & = & \int_{\hat{\Gamma}}^{\oplus}\left\langle \Phi[\varphi_{1}](\Lambda),\left((-\Delta_{\mathrm{B}}^{\Lambda}+V)\otimes1\right)\Phi[\varphi_{2}](\Lambda)\right\rangle \textrm{d}\hat{m}(\Lambda).\end{eqnarray*}
Closing the quadratic forms one arrives at equality (\ref{eq:H_int_HL}).
\end{proof}

\begin{rem} Note that clearly $\spec(\tilde{H})\subset\Sigma$. Furthermore,
it is also obvious that starting from equality (\ref{eq:H_int_HL})
one can derive a decomposition of any operator of the form $F(\tilde{H})$
where $F$ is a continuous function on $\Sigma$. For example, the
choice $F(x)=\exp(-tx)$, with $t>0$, may be of interest, and this
way one gets a decomposition of the Schr\"odinger semigroup into
a direct integral\[
\Phi\exp(-t\tilde{H})\Phi^{-1}=\int_{\hat{\Gamma}}^{\oplus}\exp(-tH^{\Lambda})\otimes1\,\mathrm{d}\hat{m}(\Lambda),\ t>0.\]
\end{rem}

\begin{rem*} The construction described earlier in \cite{kocabovastovicek}
is a particular case of the construction presented above. In more
detail, the settings in \cite{kocabovastovicek} are as follows: $\mathfrak{V}=M\times\mathbb{C}$
is the trivial line bundle over $M$, the Hermitian structure $\mathfrak{h}$
is given by the standard scalar product in $\mathbb{C}$, and $\nabla=\textrm{d}$
is the trivial flat connection. If $\pi:\tilde{M}\to M$ is a covering
map (not necessarily universal) with a structure group $\Gamma$ then
$\tilde{\mathfrak{V}}=\pi^{\ast}\mathfrak{V}=\tilde{M}\times\mathbb{C}$,
and the group $\Gamma$ acts on $\tilde{\mathfrak{V}}$ as an identity
on the fibers, $\Psi_{\gamma}(y,z)=(\gamma\cdot y,z)$ for $(y,z)\in\tilde{M}\times\mathbb{C}$
and $\gamma\in\Gamma$. \end{rem*}

\section{A formula for propagators and Green functions\label{sec:propag_periodic}}

In equation (\ref{eq:U_int_UL}), the evolution operator $\tilde{U}(t)$
is expressed in terms of $U^{\Lambda}(t)$, $\Lambda\in\hat{\Gamma}$.
It is possible to invert this relationship and to derive a formula
giving an expression for the propagator associated with $H^{\Lambda}$
in terms of the propagator associated with $\tilde{H}$. The
propagators are regarded as distributional kernels of the
corresponding evolution operators.

If $\mathfrak{W}_{1}$ and $\mathfrak{W}_{2}$ are vector fiber bundles
over manifolds $N_{1}$ and $N_{2}$, respectively, then the symbol
$\mathfrak{W}_{1}\boxtimes\mathfrak{W}_{2}$ stands for the vector
fiber bundle over $N_{1}\times N_{2}$ with fibers \[
(\mathfrak{W}_{1}\boxtimes\mathfrak{W}_{2})_{(x_{1},x_{2})}=(\mathfrak{W}_{1})_{x_{1}}\otimes(\mathfrak{W}_{2})_{x_{2}},\ (x_{1},x_{2})\in N_{1}\times N_{2}.\]

If $\mathfrak{W}$ is a vector fiber bundle over a Riemannian manifold
$N$ then the dual space to $C_{0}^{\infty}(\mathfrak{W})$ is formed
by distributional sections of the dual bundle $\mathfrak{W}^{\ast}$
(with $\mathfrak{W}_{x}^{\ast}$ being the dual space to $\mathfrak{W}_{x}$,
$x\in N$).

Suppose that $\mathfrak{W}$ is a Hermitian vector fiber bundle over
a Riemannian manifold $N$. Let \[
C^{\infty}(\mathfrak{W})\to C^{\infty}(\mathfrak{W}^{\ast}):\sigma\mapsto\overline{\sigma},\]
be the canonical antilinear isomorphism following fiber-wise from
the Riesz lemma.

The action $\Psi$ of the group $\Gamma$ on $\tilde{\mathfrak{V}}$
induces naturally an action $\Psi'$ on the dual vector fiber bundle
$\tilde{\mathfrak{V}}^{\ast}$. Analogously to (\ref{eq:def_Us})
one introduces operators $W'_{\gamma}$ on $C^{\infty}(\tilde{\mathfrak{V}}^{\ast})$,
\[
\forall\sigma\in C^{\infty}(\tilde{\mathfrak{V}}^{\ast}),\quad(W'_{\gamma}\sigma)(y)=\Psi'_{\gamma}\sigma(\gamma^{-1}\cdot y).\]
Notice that\[
\forall\varphi\in C^{\infty}(\tilde{\mathfrak{V}}),\quad W'_{\gamma}\overline{\varphi}=\overline{W_{\gamma}\varphi}.\]

The operators $W_{\gamma}$, $\gamma\in\Gamma$, defined in
(\ref{eq:def_Us}) on smooth sections can be extended to distributional
sections. Let $\alpha$ be a distributional section of
$\tilde{\mathfrak{V}}$, i.e.\ a continuous functional on
$C_{0}^{\infty}(\tilde{\mathfrak{V}}^{\ast})$.  Then, for
$\gamma\in\Gamma$,
\begin{equation}
  \forall\sigma\in C_{0}^{\infty}(\tilde{\mathfrak{V}}^{\ast}),
  \quad W_{\gamma}\alpha(\sigma)=\alpha(W'_{\gamma^{-1}}\sigma).
  \label{eq:def_Ws_gen}
\end{equation}
Similarly, let $\beta$ be a distributional section of
$\tilde{\mathfrak{V}}^{\ast}$. Then
\begin{equation}
  \forall\varphi\in C_{0}^{\infty}(\tilde{\mathfrak{V}}),
  \quad W'_{\gamma}\beta(\varphi)
  =\beta(W_{\gamma^{-1}}\varphi).\label{eq:def_Wsprim_gen}
\end{equation}
Of course, in definitions (\ref{eq:def_Ws_gen}),
(\ref{eq:def_Wsprim_gen}) one takes into account the invariance of
measure $\tilde{\mu}$ with respect to the action of $\Gamma$.

For $\varphi_{1},\varphi_{2}\in C_{0}^{\infty}(\mathfrak{W})$ denote
by $\overline{\varphi_{1}}\otimes\varphi_{2}$ the section of
$\mathfrak{W}^{\ast}\boxtimes\mathfrak{W}$ given by\linebreak
$(\overline{\varphi_{1}}\otimes\varphi_{2})(x_{1},x_{2})=\overline{\varphi_{1}}(x_{1})\otimes\varphi_{2}(x_{2})$,
$(x_{1},x_{2})\in N_{1}\times N_{2}$. Let $B$ be a bounded operator in
$L^{2}(\mathfrak{W})$. As a corollary of the Schwartz kernel theorem
(see, for example, \cite[Theorem~5.2.1]{hoermander}) one introduces
the kernel $\beta$ of $B$ as a distributional section of
$(\mathfrak{W}^{\ast}\boxtimes\mathfrak{W})^{\ast}\equiv\mathfrak{W}%
\boxtimes\mathfrak{W}^{\ast}$ that is unambiguously given by the
relation
\[
\forall\varphi_{1},\varphi_{2}\in C_{0}^{\infty}(\mathfrak{W}),
\quad\beta(\overline{\varphi_{1}}\otimes\varphi_{2})
=\langle\varphi_{1},B\varphi_{2}\rangle.
\]
The map $B\mapsto\beta$ is injective.

If $B$ is a bounded operator in $L^{2}(\tilde{\mathfrak{V}})$ then $B$
is $\Gamma$-periodic, i.e.\ $B$ commutes with all $W_{\gamma}$,
$\gamma\in\Gamma$, if and only if the distributional kernel $\beta$
satisfies the invariance condition
\begin{equation}
  \forall\gamma\in\Gamma,\ (W_{\gamma}\otimes W'_{\gamma})\beta=\beta.
  \label{eq:G_invarince_kernel}
\end{equation}

Suppose that $B$ is a bounded operator in
$L^{2}(\mathfrak{V}^{\Lambda})$, $\Lambda\in\hat{\Gamma}$. The kernel
$\beta$ of $B$ is a distributional section of
$\mathfrak{V}^{\Lambda}\boxtimes(\mathfrak{V}^{\Lambda})^{\ast}$.
Alternatively, one can characterize $B$ by a distributional section
$\beta^{\Lambda}$ of
\[
\tilde{\mathfrak{V}}^{\Lambda}
\boxtimes(\tilde{\mathfrak{V}}^{\Lambda})^{\ast}
\equiv(\tilde{\mathfrak{V}}\boxtimes\tilde{\mathfrak{V}}^{\ast})
\otimes\sI(\sL_{\Lambda})
\]
given by
\begin{eqnarray}
  &  & \hskip-2em\forall v_{1},v_{2}\in\sL_{\Lambda},
  \forall\varphi_{1},\varphi_{2}
  \in C_{0}^{\infty}(\tilde{\mathfrak{V}}),\nonumber \\
  &  & \hskip-2em\langle v_{1},\beta^{\Lambda}(\overline{\varphi_{1}}
  \otimes\varphi_{2})v_{2}\rangle_{\sL_{\Lambda}}
  =\beta\!\left(\overline{(\Phi^{\Lambda}\varphi_{1}\otimes v_{1})}
    \otimes(\Phi^{\Lambda}\varphi_{2}\otimes v_{2})\right)\!=
  \langle\Phi^{\Lambda}\varphi_{1}\otimes v_{1},B\,\Phi^{\Lambda}
  \varphi_{2}\otimes v_{2}\rangle.\nonumber \\
  &  & \mbox{{}}\label{eq:def_betaL}
\end{eqnarray}
Here we regard $\beta^{\Lambda}$ as a continuous functional on
$C_{0}^{\infty}(\tilde{\mathfrak{V}}^{\ast}\boxtimes\tilde{\mathfrak{V}})$
with values in $\sI(\sL_{\Lambda})$. Moreover, $\beta^{\Lambda}$
fulfills
\begin{equation}
  \forall\gamma\in\Gamma,\quad(W_{\gamma}\otimes1)\beta^{\Lambda}
  =\Lambda(\gamma^{-1})\beta^{\Lambda},\,\,(1\otimes W'_{\gamma})
  \beta^{\Lambda}=\beta^{\Lambda}\Lambda(\gamma).
  \label{eq:Ws_beta}
\end{equation}
Again, the map $B\mapsto\beta^{\Lambda}$ is injective.

Let $t$ be a real parameter. Denote by $\tilde{\mathcal{K}}_{t}$ the
kernel of $\tilde{U}(t)$, and by $\mathcal{K}_{t}^{\Lambda}$ the
kernel of $U^{\Lambda}(t)$. Thus $\tilde{\mathcal{K}}_{t}$ is a
distributional section of
$\tilde{\mathfrak{V}}\boxtimes\tilde{\mathfrak{V}}^{\ast}$, and
$\mathcal{K}_{t}^{\Lambda}$ is a distributional section of
$\tilde{\mathfrak{V}}\boxtimes\tilde{\mathfrak{V}}^{\ast}$ with values
in $\sI(\sL_{\Lambda})$. Moreover, the kernel
$\mathcal{K}_{t}^{\Lambda}$ is $\Lambda$-equivariant in the sense of
(\ref{eq:Ws_beta}).

Let us rewrite the Bloch decomposition of the propagator
(\ref{eq:U_int_UL}) in terms of kernels. The following lemma is a
straightforward modification of Lemma~12 and Proposition~13 in
\cite{kocabovastovicek} and so we omit the proof.

\begin{lem}
  For all $\varphi_{1},\varphi_{2}\in
  C_{0}^{\infty}(\tilde{\mathfrak{V}})$ and $\gamma\in\Gamma$, the
  function
  \[
  \Lambda\mapsto\Tr[\Lambda(\gamma)^{\ast}
  \mathcal{K}_{t}^{\Lambda}(\overline{\varphi_{1}}\otimes\varphi_{2})]
  \]
  is integrable on $\hat{\Gamma}$ and one has
  \begin{equation}
    (W_{\gamma}\otimes1)\tilde{\mathcal{K}}_{t}(\overline{\varphi_{1}}
    \otimes\varphi_{2})
    =\int_{\hat{\Gamma}}\Tr[\Lambda(\gamma)^{\ast}
    \mathcal{K}_{t}^{\Lambda}(\overline{\varphi_{1}}
    \otimes\varphi_{2})]\,\mathrm{d}\hat{m}(\Lambda).
    \label{eq:WsKt_eq_intTr}
  \end{equation}
\end{lem}

\begin{define} For $\varphi_{1},\varphi_{2}\in C_{0}^{\infty}(\tilde{\mathfrak{V}})$
arbitrary but fixed we set \[
\forall\gamma\in\Gamma,\textrm{~}F_{t}(\gamma)=(W_{\gamma}\otimes1)\tilde{\mathcal{K}}_{t}\!\left(\overline{\varphi_{1}}\otimes\varphi_{2}\right)=\tilde{\mathcal{K}}_{t}\!\left(\overline{W_{\gamma^{-1}}\varphi_{1}}\otimes\varphi_{2}\right),\]
and \[
\forall\Lambda\in\hat{\Gamma},\textrm{~}G_{t}(\Lambda)=\mathcal{K}_{t}^{\Lambda}(\overline{\varphi_{1}}\otimes\varphi_{2})\in\sI(\sL_{\Lambda})\,.\]
\end{define}

Absolutely in the same manner as in the proof of Lemma~14 in \cite{kocabovastovicek}
one can show the following lemma.

\begin{lem} \label{thm:Ft_Gt_L2} $F_{t}\in L^{2}(\Gamma)$ and $\Lambda\mapsto\|G_{t}(\Lambda)\|$
is a bounded function on $\hat{\Gamma}$. Recalling that $\hat{m}(\hat{\Gamma})\leq1$
one has $\|G_{t}(\cdot)\|\in L^{1}(\hat{\Gamma})\cap L^{2}(\hat{\Gamma})$.
In particular,\[
G_{t}\in\int_{\hat{\Gamma}}^{\oplus}\sI(\sL_{\Lambda})\,\mathrm{d}\hat{m}(\Lambda).\]
\end{lem}

In view Lemma~\ref{thm:Ft_Gt_L2}, the following proposition is an
easy corollary of (\ref{eq:WsKt_eq_intTr}).

\begin{prop} For all $\varphi_{1},\varphi_{2}\in C_{0}^{\infty}(\tilde{\mathfrak{V}})$
one has\begin{equation}
F_{t}=\sF^{-1}[G_{t}],\mbox{ }\; G_{t}=\sF[F_{t}].\label{schulman}\end{equation}
\end{prop}

\begin{rem} Rewriting formally the second equality in (\ref{schulman})
gives\begin{equation}
\mathcal{K}_{t}^{\Lambda}(y_{1},y_{2})=\sum_{\gamma\in\Gamma}\left((\Psi_{\gamma}\otimes1)\cdot\tilde{\mathcal{K}}_{t}(\gamma^{-1}\cdot y_{1},y_{2})\right)\!\otimes\Lambda(\gamma).\label{eq:formulaK}\end{equation}
Alternatively, using the fact that the Hamiltonian $\tilde{H}$ is
$\Gamma$-periodic and that the kernel $\tilde{\mathcal{K}}_{t}$
obeys the invariance condition (\ref{eq:G_invarince_kernel}) one
can write\[
\mathcal{K}_{t}^{\Lambda}(y_{1},y_{2})=\sum_{\gamma\in\Gamma}\left((1\otimes\Psi'_{\gamma})\cdot\tilde{\mathcal{K}}_{t}(y_{1},\gamma^{-1}\cdot y_{2})\right)\!\otimes\Lambda(\gamma^{-1}).\]
In the case of a flat connection formula (\ref{eq:formulaK}) coincides
with the formula for propagators on multiply connected spaces as described
in \cite{schulman1,schulman2}. \end{rem}

A formula analogous to (\ref{eq:formulaK}) can be also derived for
the corresponding Green functions. From the theoretical point of view
it can be even more convenient to work with Green functions instead
of propagators for some properties of Green functions are easier to
control than those of propagators. Let us sketch how the procedure
should be modified for this purpose.

Recall (\ref{eq:R_int_RL}). Let $z\in\mathbb{C}\setminus\Sigma$
be a spectral parameter. Denote by $\tilde{\mathcal{G}}_{z}$ the
kernel of $\tilde{R}(z)$, and by $\mathcal{G}_{z}^{\Lambda}$ the
kernel of $R^{\Lambda}(z)$. Thus $\tilde{\mathcal{G}}_{z}$ is a
distributional section of $\tilde{\mathfrak{V}}\boxtimes\tilde{\mathfrak{V}}^{\ast}$,
and $\mathcal{G}_{z}^{\Lambda}$ is a distributional section of $\tilde{\mathfrak{V}}\boxtimes\tilde{\mathfrak{V}}^{\ast}$
with values in $\sI(\sL_{\Lambda})$. Moreover, the kernel $\mathcal{G}_{z}^{\Lambda}$
is $\Lambda$-equivariant. $\tilde{\mathcal{G}}_{z}$ and $\mathcal{G}_{z}^{\Lambda}$
are the Green functions of $\tilde{H}$ and $H^{\Lambda}$, respectively.

One can again rewrite the Bloch decomposition of the Green function
(\ref{eq:R_int_RL}) in terms of kernels.

\begin{lem} For all $\varphi_{1},\varphi_{2}\in C_{0}^{\infty}(\tilde{\mathfrak{V}})$
and $\gamma\in\Gamma$, the function\[
\Lambda\mapsto\Tr[\Lambda(\gamma)^{\ast}\mathcal{G}_{z}^{\Lambda}(\overline{\varphi_{1}}\otimes\varphi_{2})]\]
is integrable on $\hat{\Gamma}$ and one has\begin{equation}
(W_{\gamma}\otimes1)\tilde{\mathcal{G}}_{z}(\overline{\varphi_{1}}\otimes\varphi_{2})=\int_{\hat{\Gamma}}\Tr[\Lambda(\gamma)^{\ast}\mathcal{G}_{z}^{\Lambda}(\overline{\varphi_{1}}\otimes\varphi_{2})]\,\mathrm{d}\hat{m}(\Lambda).\label{eq:WsRz_eq_intTr}\end{equation}
\end{lem}

\begin{proof} Let us sketch the basic steps. First one finds that\[
\langle\Phi[\varphi_{1}](\Lambda),(R^{\Lambda}(z)\otimes1)\Phi[\varphi_{2}](\Lambda)\rangle=\Tr[\mathcal{G}_{z}^{\Lambda}(\overline{\varphi_{1}}\otimes\varphi_{2})].\]
Taking into account the unitarity of $\Phi$ it also follows that\[
\int_{\hat{\Gamma}}\left|\Tr[\mathcal{G}_{z}^{\Lambda}(\overline{\varphi_{1}}\otimes\varphi_{2})]\right|\,\mathrm{d}\hat{m}(\Lambda)\leq\frac{\|\varphi_{1}\|\,\|\varphi_{2}\|}{\dist(z,\Sigma)}\,,\]
and so the function $\Lambda\mapsto\Tr[\mathcal{G}_{z}^{\Lambda}(\overline{\varphi_{1}}\otimes\varphi_{2})]$
is integrable on $\hat{\Gamma}$. Relation (\ref{eq:R_int_RL}) can
be rewritten as\[
\tilde{\mathcal{G}}_{z}(\overline{\varphi_{1}}\otimes\varphi_{2})=\int_{\hat{\Gamma}}\Tr[\mathcal{G}_{z}^{\Lambda}(\overline{\varphi_{1}}\otimes\varphi_{2})]\,\mathrm{d}\hat{m}(\Lambda),\]
and replacing $\varphi_{1}$ by $W_{\gamma^{-1}}\varphi_{1}$, $\gamma\in\Gamma$,
and taking into account the $\Lambda$-equivariance of $\mathcal{G}_{z}^{\Lambda}$
one gets (\ref{eq:WsRz_eq_intTr}). \end{proof}

\begin{define} For $\varphi_{1},\varphi_{2}\in C_{0}^{\infty}(\tilde{\mathfrak{V}})$
arbitrary but fixed put \[
\forall\gamma\in\Gamma,\textrm{~}F_{z}(\gamma)=(W_{\gamma}\otimes1)\tilde{\mathcal{G}}_{z}\!\left(\overline{\varphi_{1}}\otimes\varphi_{2}\right)=\tilde{\mathcal{G}}_{z}\!\left(\overline{W_{\gamma^{-1}}\varphi_{1}}\otimes\varphi_{2}\right),\]
and \[
\forall\Lambda\in\hat{\Gamma},\textrm{~}G_{z}(\Lambda)=\mathcal{G}_{z}^{\Lambda}(\overline{\varphi_{1}}\otimes\varphi_{2})\in\sI(\sL_{\Lambda})\,.\]
\end{define}

\begin{lem} $F_{z}\in L^{2}(\Gamma)$ and $\Lambda\mapsto\|G_{z}(\Lambda)\|$
is a bounded function on $\hat{\Gamma}$. Since $\hat{m}(\hat{\Gamma})\leq1$
one has $\|G_{z}(\cdot)\|\in L^{1}(\hat{\Gamma})\cap L^{2}(\hat{\Gamma})$.
In particular,\[
G_{z}\in\int_{\hat{\Gamma}}^{\oplus}\sI(\sL_{\Lambda})\,\mathrm{d}\hat{m}(\Lambda).\]
\end{lem}

\begin{proof} One can proceed very similarly as in the proof of Lemma~14
in \cite{kocabovastovicek}. Let us just indicate a couple of modifications.
Assuming (without loss of generality) that the sets $\gamma\cdot\supp(\varphi_{1})$,
$\gamma\in\Gamma$, are mutually disjoint one derives the estimate\[
\sum_{\gamma\in\Gamma}|F_{z}(\gamma)|^{2}\leq\|\varphi_{1}\|^{2}\|\tilde{R}(z)\varphi_{2}\|^{2}\leq\frac{\|\varphi_{1}\|^{2}\|\varphi_{2}\|^{2}}{\dist(z,\Sigma)^{2}}\,.\]
Furthermore, still assuming that the sets $\gamma\cdot\supp(\varphi_{j})$,
$\gamma\in\Gamma$, are mutually disjoint both for $j=1$ and $j=2$,
one gets\[
\|G_{z}(\Lambda)\|\leq\max_{\Lambda\in\hat{\Gamma}}\left(\dim(\sL_{\Lambda})\right)\frac{\|\varphi_{1}\|\,\|\varphi_{2}\|}{\dist(z,\Sigma)}\,.\]
The lemma follows. \end{proof}

Finally, equality (\ref{eq:WsRz_eq_intTr}) implies the following
proposition.

\begin{prop} \label{thm:FormulaG} For all $\varphi_{1},\varphi_{2}\in C_{0}^{\infty}(\tilde{\mathfrak{V}})$,
\begin{equation}
F_{z}=\sF^{-1}[G_{z}],\mbox{ }\; G_{z}=\sF[F_{z}].\label{eq:schulmanG}\end{equation}
\end{prop}

\begin{rem} The second equality in (\ref{eq:schulmanG}) can be formally
rewritten as

\begin{equation}
\mathcal{G}_{z}^{\Lambda}(y_{1},y_{2})=\sum_{\gamma\in\Gamma}\left((\Psi_{\gamma}\otimes1)\tilde{\mathcal{G}}_{z}(\gamma^{-1}\cdot y_{1},y_{2})\right)\!\otimes\Lambda(\gamma).\label{eq:formulaG}\end{equation}
And this is the formula for Green functions. Alternatively, in view
of the invariance condition (\ref{eq:G_invarince_kernel}) one can
also write\[
\mathcal{G}_{z}^{\Lambda}(y_{1},y_{2})=\sum_{\gamma\in\Gamma}\left((1\otimes\Psi'_{\gamma})\tilde{\mathcal{G}}_{z}(y_{1},\gamma^{-1}\cdot y_{2})\right)\!\otimes\Lambda(\gamma^{-1}).\]
\end{rem}

\section{Particular cases and examples}

\subsection{A reduction to the maximal Abelian covering space}

As already mentioned there are multiply connected manifolds
(configuration spaces) of interest whose fundamental group is not of
type I. A well known example is a two-dimensional Euclidean space with
two or more omitted points whose fundamental group is isomorphic to
the free group with one generator per omitted point. If one is
interested in the propagator formula (\ref{eq:formulaK}) or in the
Green function formula (\ref{eq:formulaG}) for magnetic Schr\"odinger
operators, i.e.\ one considers a situation when the gauge group is
$U(1)$, then one may avoid working with the universal covering space
$\tilde{M}$ of $M$ whose structure group is $\pi_{1}(M)$ and instead
employ the maximal Abelian covering space $\hat{M}$ whose structure
group is Abelian and isomorphic to $H_{1}(M;\mathbb{Z})$. The notion
of the maximal Abelian covering space already proved itself to be
useful in the spectral analysis of magnetic Schr\"odinger operators
\cite{sunada_maxab}. The authors are indebted to Takuya Mine for
pointing out to them this possibility.

In what follows we restrict ourselves to the formula for Green
functions since it is much easier to handle than the propagator
formula. Moreover, the covariant derivative is supposed to be flat on
$M$ and hence trivial on $\tilde{M}$, and scalar potentials are not
considered ($V=0$).

Let us recall that a maximal Abelian covering space of $M$ is a
covering space\linebreak{}%
$\hat{p}:\hat{M}\to M$ such that
$\hat{p}_{\ast}\pi_{1}(\hat{M})=[\pi_{1}(M),\pi_{1}(M)]$.  It is known
to exist and to be unique up to equivalence. This is a normal covering
space and the covering group is isomorphic to
$\pi_{1}(M)/[\pi_{1}(M),\pi_{1}(M)]=H_{1}(M;\mathbb{Z})$ (see, for
instance, \cite{turaev}). Still assuming that $M$ is a connected
Riemannian manifold, $\tilde{M}$ as well as $\hat{M}$ become
Riemannian manifolds in a unique manner so that the corresponding
projections are isometric at every point. The action of the covering
group on $\tilde{M}$ or $\hat{M}$ is then isometric, free, transitive
on the fibers and properly discontinuous \cite{lee}.

If the gauge group is $U(1)$, then only one-dimensional
representations of $\pi_{1}(M)$ are relevant. Since any
one-dimensional representation of $\pi_{1}(M)$ is trivial on
$[\pi_{1}(M),\pi_{1}(M)]$ it induces a representation of
$\pi_{1}(M)/[\pi_{1}(M),\pi_{1}(M)]=H_{1}(M;\mathbb{Z})$ - the
structure group of the covering space $\hat{M}\to M$. In that case the
universal covering space can be reduced to the maximal Abelian
covering space. Because the group $H_{1}(M;\mathbb{Z})$ is Abelian,
one may refer to Proposition~\ref{thm:FormulaG} to justify formula
(\ref{eq:formulaG}). The knowledge of the Green function
$\hat{\mathcal{G}}_{z}$ on $\hat{M}$ is required, however. Suppose one
knows the Green function $\tilde{\mathcal{G}}_{z}$ on $\tilde{M}$
rather than the Green function $\hat{\mathcal{G}}_{z}$ on $\hat{M}$.
Since $\tilde{M}\to\hat{M}$ is a covering space with the structure
group $[\pi_{1}(M),\pi_{1}(M)]$ one can formally construct
$\hat{\mathcal{G}}_{z}$ as the sum
\[
\hat{\mathcal{G}}_{z}(y,y_{0})
=\sum_{\gamma\in[\pi_{1}(M),\pi_{1}(M)]}
\tilde{\mathcal{G}}_{z}(\gamma\cdot y,y_{0}).
\]
Below we aim to verify its convergence in the sense of distributions.

We are going to consider a bit more general situation. Let $X$ be
a connected Riemannian manifold, and $\Gamma$ be a discrete symmetry
group of $X$. Suppose the action of $\Gamma$ on $X$ is isometric,
free and properly discontinuous. Let $H$ be the free Hamiltonian
in $L^{2}(X)$ introduced as the Friedrichs extension of the symmetric
positive operator $-\Delta_{\mathrm{LB}}$ (the Laplace-Beltrami operator)
with the domain $C_{0}^{\infty}(X)$. The symmetric operator is known
to be essentially selfadjoint on complete Riemannian manifolds \cite{roelcke,strichartz}.
Otherwise, in the general case, $H$ is sometimes said to be determined
by the generalized Dirichlet boundary conditions \cite{cheeger_yau},
the form domain of $H$ coincides with the Sobolev space $H_{0}^{1}(X)$,
and the closed quadratic form associated with~$H$ reads \[
H_{0}^{1}(X)\ni f\mapsto\int_{X}\mathfrak{g}(\mbox{d}f,\mbox{d}f)\,\mbox{d}\mu.\]

Denote by $R(z)=(H-z)^{-1}$, $\Re z<0$, the corresponding resolvent
restricted to the left halfplane. The free Green function $\mathcal{G}_{z}$
is a distribution on $X\times X$ defined by\[
\forall\varphi_{1},\varphi_{2}\in C_{0}^{\infty}(X),\mbox{ }\mathcal{G}_{z}(\overline{\varphi_{1}}\otimes\varphi_{2})=\langle\varphi_{1},R(z)\varphi_{2}\rangle.\]
Note that\[
\mathcal{G}_{z}(\gamma\cdot x,\gamma\cdot y)=\mathcal{G}_{z}(x,y),\ \forall\gamma\in\Gamma.\]

Put $\hat{X}=X/\Gamma$. Denote by $\hat{H}$ the free Hamiltonian
on $\hat{X}$. $\hat{H}$ is again the Friedrichs extension of minus
the Laplace-Beltrami operator with the domain $C_{0}^{\infty}(\hat{X})$.
It is convenient to identify $L^{2}(\hat{X})$ with the Hilbert space
$\hat{\sH}$ formed by $\Gamma$-periodic functions on $X$ which
are $L^{2}$ integrable over a fundamental domain. Then, as a differential
operator, $\hat{H}$ coincides with $-\Delta_{\mathrm{LB}}$. Put
$\hat{R}(z)=(\hat{H}-z)^{-1}$, $\Re z<0$.

Again, $C_{\Gamma}^{\infty}(X)$ stands for the vector space of smooth
$\Gamma$-periodic functions on $X$, and we define\begin{equation}
\Phi_{0}:C_{0}^{\infty}(X)\to C_{\Gamma}^{\infty}(X):\varphi\mapsto\sum_{\gamma\in\Gamma}L_{\gamma}^{\ast}\varphi.\label{eq:def_Phi}\end{equation}
If $\varphi\in C_{0}^{\infty}(X)$, then $\Phi_{0}\varphi=\hat{p}^{\ast}\sigma$
for a unique $\sigma\in C_{0}^{\infty}(\hat{X})$. This defines a
linear mapping $C_{0}^{\infty}(X)\to C_{0}^{\infty}(\hat{X})$ which
is surjective (see Lemma~\ref{thm:PhiL_surject}). Particularly,
it follows that \[
\Phi_{0}\!\left(C_{0}^{\infty}(X)\right)\subset\Dom(\hat{H})\subset\hat{\sH}\]
is a core of $\hat{H}$.

The free Green function $\hat{\mathcal{G}}_{z}$ on $\hat{X}$
(associated with $\hat{H}$) is identified with a
$\Gamma\times\Gamma$-periodic distribution on $X\times X$
unambiguously determined by the relation
\[
\forall\varphi_{1},\varphi_{2}\in C_{0}^{\infty}(X),\mbox{ }
\hat{\mathcal{G}}_{z}(\overline{\varphi_{1}}\otimes\varphi_{2})
=\langle\Phi_{0}\varphi_{1},\hat{R}(z)\Phi_{0}\varphi_{2}\rangle
\]
(the scalar product is taken in $\hat{\sH}$).

Consider the semigroup $\exp(-tH)$, $t>0$. The corresponding
distributional kernel $p(t;x,y)$ is the heat kernel on $X$ (for the
generalized Dirichlet boundary conditions). The heat kernel $p(t;x,y)$
is known to be a smooth and strictly positive function on
$(0,+\infty)\times X\times X$ which is symmetric in the variables $x$
and $y$. It is unambiguously characterized as the smallest positive
fundamental solution of the heat equation on $X$. Moreover, one has
\begin{equation}
  \forall x\in X,\mbox{ }\int_{X}p(t;x,y)\,\mbox{d}\mu(y) \leq 1
  \label{eq:int_p_eq1}
\end{equation}
\cite{dodziuk}. Under certain assumptions it is even true that
inequality (\ref{eq:int_p_eq1}) becomes in fact an equality for any
$x\in{}X$ \cite{grigoryan}, for example when $X$ is a complete
Riemannian manifold of Ricci curvature bounded from below
\cite{dodziuk}. This is consistent with the probabilistic
interpretation -- for a fixed $x$, $p(t;x,y)$ is the probability
density of a diffusion process when a particle departs from the point
$x$ at time $0$ and reaches a variable point $y$ at time $t>0$
\cite{hunt}. For our purposes inequality (\ref{eq:int_p_eq1}) is
pretty sufficient, however.

Note that the Green function equals the Laplace transform of the heat
kernel,
\begin{equation}
  \mathcal{G}_{z}(x,y)=\int_{0}^{\infty}e^{zt}p(t;x,y)\,
  \mbox{d}t,\ \Re z<0.
  \label{eq:Gz_Laplacetr_p}
\end{equation}
This also means that $\mathcal{G}_{z}(x,y)$ is a regular distribution
and $R(z)$ is an integral operator.

\begin{lem}
  \label{thm:Green_L1}
  On a general Riemannian manifold $X$ and for all $z\in\mathbb{C}$
  with $\Re z<0$, the integral operator $R(z)$ is bounded on
  $L^{1}(X)$ with the upper bound $1/|\Re z|$. In particular,
  $\forall\varphi\in C_{0}^{\infty}(X)$, $R(z)\varphi\in L^{1}(X)$.
\end{lem}

\begin{proof}
  From (\ref{eq:int_p_eq1}) and (\ref{eq:Gz_Laplacetr_p}) it follows
  that, for all $\varphi\in L^{1}(X)$,
  \[
  \|R(z)\varphi\|_{1}
  \leq\int_{X}\int_{X}|\mathcal{G}_{z}(x,y)||\varphi(y)|\,
  \mbox{d}\mu(x)\mbox{d}\mu(y)\leq\frac{\|\varphi\|_{1}}{|\Re z|}\,.
  \]
  This proves the lemma.
\end{proof}

\begin{lem}
  Suppose $D\subset X$ is a fundamental domain of the action of
  $\Gamma$, $\varphi\in C_{0}^{\infty}(X)$. Then
  \begin{equation}
    \|\Phi_{0}R(z)\varphi\Big|_{D}\|_{1}
    \leq\frac{\|\varphi\|_{1}}{|\Re z|}\,,\ \|\Phi_{0}R(z)\varphi
    \Big|_{D}\|_{\infty}\leq C_{\supp(\varphi)}
    \frac{\|\varphi\|_{\infty}}{|\Re z|}\,,
    \label{eq:PhiRphi_ineqs}
  \end{equation}
  where $C_{\supp(\varphi)}\geq0$ depends only on $\supp(\varphi)$.
  Consequently,
  \begin{equation}
    \|\Phi_{0}R(z)\varphi\Big|_{D}\|_{2}
    \leq\frac{1}{|\Re z|}\,\sqrt{C_{\supp(\varphi)}\|\varphi\|_{1}
      \|\varphi\|_{\infty}}\label{eq:PhiRphi_ineq2}
  \end{equation}
  (the norms of the restrictions to the domain $D$ occurring on the
  left hand sides are taken in $L^{p}(D)$, and otherwise the norms are
  taken in $L^{p}(X)$ for an appropriate $p$).
\end{lem}

\begin{proof}
  First observe that $f\in L^{1}(X)$ implies $\Phi_{0}f\in L^{1}(D)$
  and $\|\Phi_{0}f\big|_{D}\|_{1}\leq\|f\|_{1}$ where, similarly to
  (\ref{eq:def_Phi}),
  $\Phi_{0}f(x)=\sum_{\gamma\in\Gamma}f(\gamma\cdot x)$.  In fact, one
  has
  \[
  \|\Phi_{0}f\Big|_{D}\|_{1}
  =\int_{D}\Big|\sum_{\gamma\in\Gamma}
  f(\gamma\cdot x)\Big|\,\mbox{d}\mu(x)
  \leq\,\sum_{\gamma\in\Gamma}\int_{\gamma\cdot D}
  |f(x)|\,\mbox{d}\mu(x)=\|f\|_{1}.
  \]
  Now, to get the first inequality in (\ref{eq:PhiRphi_ineqs}) it
  suffices to apply Lemma \ref{thm:Green_L1}. Further, one can assume,
  without loss of generality, that the sets
  $\gamma\cdot\supp(\varphi)$, $\gamma\in\Gamma$, are mutually
  disjoint (in that case $C_{\supp(\varphi)}=1$). Then
  \[
  |\Phi_{0}R(z)\varphi(x)|\leq\|\varphi\|_{\infty}
  \sum_{\gamma\in\Gamma}\int_{\gamma^{-1}\cdot\supp(\varphi)}
  \mathcal{G}_{\Re(z)}(x,y)\,\mbox{d}\mu(y)
  \leq\frac{\|\varphi\|_{\infty}}{|\Re z|}\,.
  \]
  Finally, one has $\|f\|_{2}^{\,2}\leq\|f\|_{1}\|f\|_{\infty}$.
\end{proof}

\begin{prop}
  The equality
  \begin{equation}
    \hat{\mathcal{G}}_{z}(x_{1},x_{2})
    =\sum_{\gamma\in\Gamma}\mathcal{G}_{z}(\gamma\cdot x_{1},x_{2})
    \label{eq:sum_Gzg}
  \end{equation}
  holds on $X\times X$ in the sense of distributions.
\end{prop}

\begin{proof}
  To show that series (\ref{eq:sum_Gzg}) converges in the sense of
  distributions it suffices to verify its convergence on every test
  function $\varphi\in C_{0}^{\infty}(X\times X)$. Such a
  simplification is possible in view of the definition of the topology
  in the distribution space and, in particular, owing to the
  completeness of this space \cite[\S5.3]{vladimirov}. Choose
  $\varphi\in{}C_{0}^{\infty}(X\times X)$. Then
  $\supp\varphi\subset{}K_{1}\times K_{2}$ for some compact subsets
  $K_{1},K_{2}\subset X$. There exist open sets $U_{j}\subset X$,
  $1\leq j\leq N$, so that $K_{1}\subset\cup U_{j}$ and the sets
  $\gamma\cdot U_{j}$, $\gamma\in\Gamma$, are mutually disjoint for
  every $j$. One has, for $\Re z<0,$
\begin{eqnarray*}
  \sum_{\gamma\in\Gamma}\left|\mathcal{G}_{z}(\gamma\cdot x_{1},x_{2})
    \big(\varphi(x_{1},x_{2})\big)\right|
  & \leq & \|\varphi\|_{\infty}\sum_{\gamma\in\Gamma}
  \int_{K_{1}\times K_{2}}
  \mathcal{G}_{\Re(z)}(\gamma\cdot x_{1},x_{2})\,
  \mbox{d}\mu(x_{1})\mbox{d}\mu(x_{2})\\
  & \leq & \|\varphi\|_{\infty}\sum_{j=1}^{N}
  \sum_{\gamma\in\Gamma}\int_{\gamma\cdot U_{j}
    \times K_{2}}\mathcal{G}_{\Re(z)}(x_{1},x_{2})\,
  \mbox{d}\mu(x_{1})\mbox{d}\mu(x_{2})\\
  & \leq & \frac{N\|\varphi\|_{\infty}}{|\Re z|}\,\mu(K_{2}).
\end{eqnarray*}

Let us denote $\mathcal{S}_{z}(x_{1},x_{2})=\sum_{\gamma\in\Gamma}%
\mathcal{G}_{z}(\gamma\cdot{}x_{1},x_{2})\in\sD'(X\times X)$. In
particular, one has
\[
\forall\varphi_{1},\varphi_{2}\in C_{0}^{\infty}(X),
\mbox{ }\mathcal{S}_{z}(\varphi_{1}\otimes\varphi_{2})
=\sum_{\gamma\in\Gamma}\mathcal{G}_{z}(L_{\gamma}^{\ast}\varphi_{1}
\otimes\varphi_{2}).
\]
Showing that $\hat{\mathcal{G}}_{z}=\mathcal{S}_{z}$ means proving the
equality
\[
\forall\varphi_{1},\varphi_{2}\in C_{0}^{\infty}(X),
\mbox{
}\langle\Phi_{0}\varphi_{1},\hat{R}(z)\Phi_{0}\varphi_{2}\rangle
=\int_{X}\overline{\Phi_0\varphi_{1}}\, R(z)\varphi_{2}\,\mbox{d}\mu
\]
(the integral makes sense since
$\|\Phi_{0}\varphi_{1}\|_{\infty}<\infty$,
$\|R(z)\varphi_{2}\|_{1}<\infty$). This is further equivalent to
\[
\forall\varphi\in C_{0}^{\infty}(X),\mbox{ }\hat{R}(z)\Phi_{0}\varphi
=\Phi_{0}R(z)\varphi.
\]
To verify this relation one has to show that
\begin{equation}
  \forall\varphi\in C_{0}^{\infty}(X),\mbox{ }\Phi_{0}R(z)\varphi
  \in\Dom(\hat{H})\mbox{ and }(\hat{H}-z)\Phi_{0}R(z)\varphi
  =\Phi_{0}\varphi.\label{eq:HhatPhiRzphi}
\end{equation}

Let us prove (\ref{eq:HhatPhiRzphi}). From (\ref{eq:PhiRphi_ineq2}) it
follows that $\Phi_{0}R(z)\varphi$ is a $\Gamma$-periodic measurable
function on $X$ which is $L^{2}$ integrable over a fundamental domain;
so it belongs to the Hilbert space $\hat{\sH}$. Since the series
$\sum_{\gamma\in\Gamma}R(z)\varphi(\gamma\cdot x)$ converges in the
$L^{1}$ norm over any compact subset of $X$ it converges in the sense
of distributions. Hence, in the sense of distributions,
\[
(-\Delta_{\mathrm{LB}}-z)\sum_{\gamma\in\Gamma}R(z)
\varphi(\gamma\cdot x)=\sum_{\gamma\in\Gamma}
\left((-\Delta_{\mathrm{LB}}-z)R(z)\varphi\right)(\gamma\cdot x)
=\sum_{\gamma\in\Gamma}\varphi(\gamma\cdot x)=\Phi_{0}\varphi(x),
\]
i.e.\ $(-\Delta_{\mathrm{LB}}-z)\Phi_{0}R(z)\varphi=\Phi_{0}\varphi$.
Since $-\Delta_{\mathrm{LB}}-z$ is an elliptic operator with smooth
coefficients and $\Phi_{0}\varphi$ is a smooth function, both of them
on $X$, by the elliptic regularity theorem, $\Phi_{0}R(z)\varphi$ is
smooth as well (see, for example, \cite[Appendix~4~\S5]{lang}). By the
Green formula,
\[
\forall\varphi,\psi\in C_{0}^{\infty}(X),\mbox{ }
\big\langle(\hat{H}-\bar{z})\Phi_{0}\psi,\Phi_{0}R(z)\varphi\big\rangle
=\langle\Phi_{0}\psi,\Phi_{0}\varphi\rangle.
\]
This equality implies (\ref{eq:HhatPhiRzphi}) since
$\Phi_{0}\big(C_{0}^{\infty}(X)\big)$ is a core of $\hat{H}$.
\end{proof}

\subsection{The case of a trivial line bundle over $M$ }

The Green function formula (\ref{eq:formulaG}) (or the propagator
formula (\ref{eq:formulaK})) essentially simplifies in the particular
case when the gauge group is $U(1)$ and the line bundle over $M$ is
trivial. Then the line bundle over $\tilde{M}$ is trivial as well. On
the other hand, the connection need not be flat. Let us shortly
indicate the basic modifications.

Suppose $\mathfrak{V}=M\times\mathbb{C}$ with a Hermitian structure
$\mathfrak{h}$ given by the standard scalar product in $\mathbb{C}$,
and with a connection $\nabla=\textrm{d}+\alpha$ where $\alpha$ is a
one-form on $M$ with values in $i\mathbb{R}$. If $\pi:\tilde{M}\to M$
is a covering space with a structure group $\Gamma$ then
$\tilde{\mathfrak{V}}=\pi^{\ast}\mathfrak{V}=\tilde{M}\times\mathbb{C}$.
The group $\Gamma$ acts on $\tilde{\mathfrak{V}}$ as an identity on
the fibers, $\Psi_{\gamma}(y,z)=(\gamma\cdot y,z)$ for
$(y,z)\in\tilde{M}\times\mathbb{C}$ and $\gamma\in\Gamma$.
Furthermore, $\tilde{\nabla}=\mbox{d}+\tilde{\alpha}$,
$\tilde{\alpha}=\pi^{\ast}\alpha$. Let us denote by
$L_{\gamma}^{\ast}$ the pullback mapping for the left action of
$\gamma\in\Gamma$ on $\tilde{M}$. Note that the invariance condition
(\ref{eq:inv_nabla_tild}) for the connection $\tilde{\nabla}$ means
that $L_{\gamma}^{\ast}\tilde{\alpha}=\tilde{\alpha}$,
$\forall\gamma\in\Gamma$, and this is clearly satisfied. If $\Lambda$
is a one-dimensional unitary representation of $\Gamma$ then
$\tilde{\mathfrak{V}}^{\Lambda}=\tilde{M}\times\mathbb{C}$,
$\tilde{\mathfrak{h}}^{\Lambda}=\tilde{\mathfrak{h}}$,
$\tilde{\nabla}^{\Lambda}=\tilde{\nabla}$. Note, however, that
$\Psi_{\gamma}^{\Lambda}(y,z)=(\gamma\cdot y,\Lambda(\gamma)z)$.

Suppose $\eta$ is a nowhere-vanishing complex function on $\tilde{M}$
such that $\forall\gamma\in\Gamma$,
$L_{\gamma}^{\ast}\eta=\Lambda(\gamma)\eta$. Clearly, $|\eta|$ is
$\Gamma$-periodic. Replacing $\eta$ by $\eta/|\eta|$ one can assume
that $|\eta|\equiv1$ on $\tilde{M}$. Then $\eta$ defines a
trivialization of $\mathfrak{V}^{\Lambda}$. In that case
$\mathfrak{V}^{\Lambda}=M\times\mathbb{C}$,
$\mathfrak{h}^{\Lambda}=\mathfrak{h}$ and
$\alpha^{\Lambda}=\alpha+(\mbox{d}\eta)\eta^{-1}$ (more precisely,
$(\mbox{d}\eta)\eta^{-1}$ is a $\Gamma$-invariant one-form on
$\tilde{M}$ which projects to a one-form on $M$).

One also identifies $L^{2}(\tilde{\mathfrak{V}})=L^{2}(\tilde{M})$,
$L^{2}(\mathfrak{V}^{\Lambda})=L^{2}(M)$. Let
$\Delta_{\mathrm{B}}^{\Lambda}$ be the Bochner Laplacian corresponding
to the connection $\nabla^{\Lambda}=\mbox{d}+\alpha^{\Lambda}$. Thus
we consider $H^{\Lambda}=-\Delta_{\mathrm{B}}^{\Lambda}+V$ as an
operator in $L^{2}(M)$ rather than in the Hilbert space of
$\Lambda$-equivariant functions on $\tilde{M}$. Let us denote the
latter space by $\sH^{\Lambda}$.  The two spaces are related by the
unitary mapping
\begin{equation}
  L^{2}(M)\to\sH^{\Lambda}:\psi\mapsto(\pi^{\ast}\psi)\eta.
  \label{eq:unitary_L2M_HL}
\end{equation}
With this trivialization, it is natural to define the Green function
$\mathcal{G}_{z}^{\Lambda}$ of $H^{\Lambda}$ in the standard manner as
a distribution on $M\times M$,
\[
\forall\psi_{1},\psi_{2}\in C_{0}^{\infty}(M),\mbox{ }
\mathcal{G}_{z}^{\Lambda}(\overline{\psi_{1}}\otimes\psi_{2})
=\langle\psi_{1},R^{\Lambda}(z)\psi_{2}\rangle.
\]
One has to keep in mind that $\mathcal{G}_{z}^{\Lambda}$ depends on
the choice of $\eta$. This definition differs from that given in
(\ref{eq:def_betaL}) and so one has to somewhat modify the Green
function formula (\ref{eq:formulaG}). To this end, it is useful to
identify distributions on $M$ with $\Gamma$-periodic distributions on
$\tilde{M}$ as explained in the following remark.

\begin{rem} Denote by $C_{\Gamma}^{\infty}(\tilde{M})$ the vector
space of $\Gamma$-periodic smooth functions on $\tilde{M.}$ Let
us define \[
\Phi_{0}:C_{0}^{\infty}(\tilde{M})\to C_{\Gamma}^{\infty}(\tilde{M}):\varphi\mapsto\sum_{\gamma\in\Gamma}L_{\gamma}^{\ast}\varphi.\]
For $f\in\sD'(M)$ one introduces $\pi^{\ast}f\in\sD'(\tilde{M})$
by the prescription\[
\forall\varphi\in C_{0}^{\infty}(\tilde{M}),\mbox{ }\pi^{\ast}f(\varphi)=f(\psi)\mbox{ where }\pi^{\ast}\psi=\Phi_{0}\varphi.\]
The functional $\pi^{\ast}f$ is well defined on $C_{0}^{\infty}(\tilde{M})$,
and it is not difficult to verify that it is continuous. Hence $\pi^{\ast}f\in\sD'(\tilde{M})$.
Moreover, $\pi^{\ast}f$ is $\Gamma$-periodic.

Conversely, suppose $\tilde{g}\in\sD'(\tilde{M})$ is $\Gamma$-periodic.
Note that for every $\psi\in C_{0}^{\infty}(M)$ there exists $\varphi\in C_{0}^{\infty}(\tilde{M})$
such that $\Phi_{0}\varphi=\pi^{\ast}\psi$; $\varphi$ is, however,
ambiguous. On the other hand, if for some $\varphi\in C_{0}^{\infty}(\tilde{M})$,
$\Phi_{0}\varphi=0$ then $\tilde{g}(\varphi)=0$. In fact, there
exists $\chi\in C_{0}^{\infty}(\tilde{M})$ such that $\Phi_{0}\chi\equiv1$
on a neighborhood of $\supp\varphi$. Since $\tilde{g}$ is $\Gamma$-invariant
one has \[
\tilde{g}(\varphi)=\tilde{g}(\varphi\Phi_{0}\chi)=\tilde{g}(\chi\Phi_{0}\varphi)=0.\]
Let us define a functional $g$ on $C_{0}^{\infty}(M)$,\[
\forall\psi\in C_{0}^{\infty}(M),\mbox{ }g(\psi)=\tilde{g}(\varphi)\mbox{ where }\varphi\in C_{0}^{\infty}(\tilde{M})\mbox{ is s.t. }\Phi_{0}\varphi=\pi^{\ast}\psi.\]
Then $g$ is well defined and continuous. One observes that $g$ is
the only distribution on $M$ satisfying $\pi^{\ast}g=\tilde{g}$.

One concludes that the mapping $\sD'(M)\to\sD'(\tilde{M}):f\mapsto\pi^{\ast}f$
induces an isomorphism of $\sD'(M)$ onto the space of $\Gamma$-periodic
distributions on $\tilde{M}$. \end{rem}

Let us regard $\mathcal{G}_{z}^{\Lambda}$ as a
$\Gamma\times\Gamma$-periodic distribution on
$\tilde{M}\times\tilde{M}$. Using the unitary map
(\ref{eq:unitary_L2M_HL}) it is a matter of straightforward
manipulations to show that the Green function formula now reads
\begin{equation}
  \mathcal{G}_{z}^{\Lambda}(y_{1},y_{2})
  =\overline{\eta(y_{1})}\eta(y_{2})
  \sum_{\gamma\in\Gamma}\Lambda(\gamma)
  \tilde{\mathcal{G}}_{z}(\gamma^{-1}\cdot y_{1},y_{2}).
  \label{eq:formulaGpart}
\end{equation}

\subsection{The case of a trivializable line bundle over $\tilde{M}$ \label{sec:triv_line_Mtilde}}

Next we are going to discuss a still rather particular case but more
general than in the preceding subsection. The gauge group is again
supposed to be $U(1)$. The line bundle over $M$ need not be trivial
neither is the connection required to be flat. We assume, however,
that after a pull-back one gets a line bundle over the covering space
$\tilde{M}$ which is trivializable. As is well known, this surely
happens if $H^{2}(\tilde{M},\mathrm{\mathbb{Z}})=0$.

The assumption that the line bundle $\tilde{\mathfrak{V}}=\pi^{\ast}\mathfrak{V}$
be trivializable means that there exists a nowhere vanishing smooth
section $\eta\in C^{\infty}(\tilde{\mathfrak{V}})$. Without loss
of generality we assume that $\tilde{\mathfrak{h}}(\eta,\eta)=1$.
Using $\eta$ one passes from $\tilde{\mathfrak{V}}$ to the trivial
line bundle $\tilde{M}\times\mathbb{C}$. In particular, one has the
isomorphism\begin{equation}
C^{\infty}(\tilde{M})\to C^{\infty}(\tilde{\mathfrak{V}}):\varphi\mapsto\varphi\eta.\label{eq:Cinf_Vtilde_triv}\end{equation}
The Hermitian structure in $\tilde{M}\times\mathbb{C}$ is given by
the standard scalar product in $\mathbb{C}$, the covariant derivative
$\tilde{\nabla}$ becomes $\mbox{d}+\tilde{\alpha}$ where $\tilde{\alpha}$
is a one-form on $\tilde{M}$ with values in $i\mathbb{R}$.

Given $y\in\tilde{M}$ one can compare the values $\eta(y)$ and $\eta(\gamma\cdot y)$
for any $\gamma\in\Gamma$ since $\tilde{\mathfrak{V}}_{y}=\tilde{\mathfrak{V}}_{\gamma\cdot y}=\mathfrak{V}_{\pi(y)}$.
Abusing somewhat the previously used notation let us write\[
\forall\gamma\in\Gamma,\forall y\in\tilde{M},\ \eta(\gamma\cdot y)=\Psi_{\gamma}(y)^{-1}\eta(y),\]
where $\Psi_{\gamma}(y)\in\mathbb{C}$ is unambiguously determined
by this equality. Then for any $\gamma\in\Gamma$, $\Psi_{\gamma}$
is a smooth complex function on $\tilde{M}$ with values in the unit
circle. The fiber-wise linear action of $\Gamma$ on $\tilde{M}\times\mathbb{C}$
takes the form\[
\tilde{M}\times\mathbb{C}\to\tilde{M}\times\mathbb{C}:\,(y,z)\mapsto\big(\gamma\cdot y,\Psi_{\gamma}(y)z\big),\ \mbox{with\ }\gamma\in\Gamma.\]
The composition rule for this action means that\[
\forall\gamma_{1},\gamma_{2}\in\Gamma,\ \big(L_{\gamma_{2}}^{\,\ast}\Psi_{\gamma_{1}}\big)\Psi_{\gamma_{2}}=\Psi_{\gamma_{1}\gamma_{2}}.\]
Here again, $L_{\gamma}^{\ast}$ stands for the pull-back mapping
of the left action of $\gamma\in\Gamma$ on $\tilde{M}$.

Sections of the line bundle $\mathfrak{V}$ over $M$ can be naturally
identified with $\Gamma$-invariant sections of the line bundle
$\tilde{\mathbb{\mathbb{\mathfrak{V}}}}$ over $\tilde{M}$. Employing
the isomorphism (\ref{eq:Cinf_Vtilde_triv}) this means that smooth
(measurable) sections of $\mathfrak{V}$ are identified with those
smooth (measurable) functions $\varphi$ on $\tilde{M}$ which satisfy
everywhere (almost everywhere) the condition
\[
\forall\gamma\in\Gamma,\ \varphi(\gamma\cdot y)
=\Psi_{\gamma}(y)\varphi(y),
\]
i.e.\ $L_{\gamma}^{\,\ast}\varphi=\Psi_{\gamma}\varphi$. The operators
$W_{\gamma}$ defined in (\ref{eq:def_Us}) for the general case now act
either in $C^{\infty}(\tilde{M})$ or as unitary operators in
$L^{2}(\tilde{M})$ according to the rule
\begin{equation}
  \forall\gamma\in\Gamma,\ W_{\gamma}\varphi(y)
  =\Psi_{\gamma}(\gamma^{-1}\cdot y)\varphi(\gamma^{-1}\cdot y).
  \label{eq:Wgamma_magnet_transl}
\end{equation}
Thus sections of $\mathfrak{V}$ correspond exactly to those functions
on $\tilde{M}$ which satisfy $W_{\gamma}\varphi=\varphi$,
$\forall\gamma\in\Gamma$.

From the physical point of view it is of interest to observe that now
the invariance condition (\ref{eq:inv_nabla_tild}) for the connection
$\tilde{\nabla}$ in general does not mean that
$L_{\gamma}^{\ast}\tilde{\alpha}=\tilde{\alpha}$ for
$\gamma\in\Gamma$. A straightforward computation gives the correct
invariance condition for this case, namely
\[
\forall\gamma\in\Gamma,\ L_{\gamma^{-1}}^{\,\ast}\tilde{\alpha}
=\tilde{\alpha}+\Psi_{\gamma}^{\,-1}\mbox{d}\Psi_{\gamma}.
\]
On the other hand, the curvature
$\tilde{\Omega}=\mbox{d}\tilde{\alpha}$ or, in other words, the two
form of the magnetic field (up to a constant multiplier) fulfills
$L_{\gamma}^{\ast}\tilde{\Omega}=\tilde{\Omega}$,
$\forall\gamma\in\Gamma$. In the analysis of physical systems with an
invariant nonzero magnetic flux the so called magnetic translations
turn out to be a useful tool \cite{zak1,zak2}. In our formalism the
magnetic translations coincide with the operators $W_{\gamma}$ defined
in (\ref{eq:Wgamma_magnet_transl}). Particularly, the symmetry of the
Hamiltonian $\tilde{H}=-\Delta_{B}+\tilde{V}$ is reflected by the fact
that it commutes with all $W_{\gamma}$, $\gamma\in\Gamma$.

Let $\Lambda$ be a one-dimensional unitary representation of $\Gamma$.
The line bundle $\tilde{\mathfrak{V}}^{\Lambda}$ is again identified
with $\tilde{M}\times\mathbb{C}$, and the Hermitian structure and the
covariant derivative remain unchanged. What is modified, however, is
the action of $\Gamma$ on $\tilde{M}\times\mathbb{C}$. The modified
action was called $\Psi_{\gamma}^{\Lambda}$ in
Subsection~\ref{sec:assoc_vect_bundles} and now it takes the form
\[
\tilde{M}\times\mathbb{C}\to\tilde{M}\times\mathbb{C}:
(y,z)\mapsto\big(\gamma\cdot y,\Lambda(\gamma)
\Psi_{\gamma}(y)z\big),\ \mbox{with\ }\gamma\in\Gamma.
\]
Sections of $\mathfrak{V}^{\Lambda}$ are again identified with functions
$\varphi$ on $\tilde{M}$ fulfilling
\begin{equation}
  \forall\gamma\in\Gamma,\ \varphi(\gamma\cdot y)
  =\Lambda(\gamma)\Psi_{\gamma}(y)\varphi(y).
  \label{eq:L_equivar_Psi}
\end{equation}
This identification implies an isomorphism of
$L^{2}(\mathfrak{V}^{\Lambda})$ with the Hilbert space of
$\Lambda$-equivariant functions on $\tilde{M}$, i.e.\ with the Hilbert
space $\sH^{\Lambda}$ formed by measurable functions on $\tilde{M}$
satisfying (\ref{eq:L_equivar_Psi}) almost everywhere and being square
integrable over a fundamental domain of the action of $\Gamma$. Then
the formal differential expression corresponding to the Hamiltonian
$\tilde{H}$ acting in $L^{2}(\tilde{M})$ is the same as that for
$H^{\Lambda}$ acting in $\sH^{\Lambda}$.

The Green function of $\tilde{H}$ is a distribution $\tilde{\mathcal{G}}_{z}$
on $\tilde{M}\times\tilde{M}$. According to (\ref{eq:G_invarince_kernel}),
it has the symmetry property\[
\forall\gamma\in\Gamma,\ \tilde{\mathcal{G}}_{z}(\gamma\cdot y_{1},\gamma\cdot y_{2})=\Psi_{\gamma}(y_{1})\overline{\Psi_{\gamma}(y_{2})}\,\tilde{\mathcal{G}}_{z}(y_{1},y_{2}).\]
The Green function $\mathcal{G}_{z}^{\Lambda}$ of $H^{\Lambda}$
is a distribution on $\tilde{M}\times\tilde{M}$ as well. According
to (\ref{eq:Ws_beta}) it is $\Lambda$-equivariant, and this property
now reads\[
\forall\gamma\in\Gamma,\ \mathcal{G}_{z}^{\Lambda}(\gamma\cdot y_{1},y_{2})=\Lambda(\gamma)\Psi_{\gamma}(y_{1})\mathcal{G}_{z}^{\Lambda}(y_{1},y_{2}),\,\mathcal{G}_{z}^{\Lambda}(y_{1},\gamma\cdot y_{2})=\Lambda(\gamma^{-1})\overline{\Psi_{\gamma}(y_{2})}\,\mathcal{G}_{z}^{\Lambda}(y_{1},y_{2}).\]
Finally, in this example the Green function formula (\ref{eq:formulaG})
takes the form\begin{eqnarray}
\mathcal{G}_{z}^{\Lambda}(y_{1},y_{2}) & = & \sum_{\gamma\in\Gamma}\,\Lambda(\gamma)\Psi_{\gamma}(\gamma^{-1}\cdot y_{1})\tilde{\mathcal{G}}_{z}(\gamma^{-1}\cdot y_{1},y_{2})\nonumber \\
 & = & \sum_{\gamma\in\Gamma}\,\Lambda(\gamma^{-1})\overline{\Psi_{\gamma}(\gamma^{-1}\cdot y_{2})}\,\tilde{\mathcal{G}}_{z}(y_{1},\gamma^{-1}\cdot y_{2}).\label{eq:formG_trivializable}\end{eqnarray}

\subsection{A constant magnetic field on a torus}

As an illustration of the formalism described in
Subsection~\ref{sec:triv_line_Mtilde} let us consider an example in
which $M$ is the two-dimensional torus
$\mathbb{T}^{2}=(\mathbb{R}/2\pi\mathbb{Z})^{2}$. This example shares
the basic nontrivial features of the approach and, at the same time,
it is sufficiently simple to allow for explicit computations. The
Riemannian structure on $\mathbb{T}^{2}$ is induced by the standard
scalar product in $\mathbb{R}^{2}$. Let $\tilde{M}$ be the universal
covering space of the torus, i.e.\ $\tilde{M}=\mathbb{R}^{2}$. Hence
the covering group is $\Gamma=(2\pi\mathbb{Z})^{2}$. The gauge group
is supposed to be $U(1)$ and thus only line bundles are considered.
We do not require that line bundles over $M$ be trivial. For
convenience we choose a connection in such a way that the magnetic
field (curvature) is constant in the natural coordinate system on the
torus. Since any line bundle over $\mathbb{R}^{2}$ is trivializable
one can actually employ the notation and the formalism introduced in
Subsection~\ref{sec:triv_line_Mtilde}.

The magnetic Bloch analysis on a torus has already been discussed in
detail in \cite{aschetal} for the case when there is no scalar
potential. Here we focus on the inverse procedure, i.e.\ on the
formula for Green functions (\ref{eq:formulaG}). As a slight
modification if compared to \cite{aschetal}, we prefer to use the
Landau gauge rather than the symmetric one. The former gauge is also
frequently used in the physical literature dedicated to quantum
systems describing a particle on a torus in a constant magnetic field
\cite{zainuddin,alhashimiwiese}.

The standard coordinates on $\mathbb{R}^{2}$ are denoted $(x,y)$ while
the angle coordinates on $\mathbb{T}^{2}$ are denoted $(\phi,\theta)$.
The projection $\pi:\mathbb{R}^{2}\to\mathbb{T}^{2}$ is given by
\[
\phi=x\,(\mbox{mod}\,2\pi),\ \theta=y\,(\mbox{mod}\,2\pi).
\]
Since $H^{2}(\mathbb{T}^{2},\mathbb{Z})=\mathbb{Z}$, equivalence
classes of line bundles over $\mathbb{T}^{2}$ are labeled by integers.
To describe a line bundle over $\mathbb{T}^{2}$ in terms of transition
functions we cover the torus by two cylinders:
 \[
U_{1}=(0,2\pi)\times\mathbb{T}^{1},\ U_{2}
=(-\pi,\pi)\times\mathbb{T}^{1}.
\]
The coordinates on $U_{1}$ and $U_{2}$ are $(\phi_{1},\theta)$ and
$(\phi_{2},\theta)$, respectively, and so the coordinate $\theta$ is
the same for both cylinders. The intersection $U_{1}\cap U_{2}$ equals
the disjoint union $U_{12A}\cup U_{12B}$ where $U_{12A}$ is
$(0,\pi)\times\mathbb{T}^{1}$ both in the coordinates
$(\phi_{1},\theta)$ and $(\phi_{2},\theta)$ while $U_{12B}$ is
identified either with $(\pi,2\pi)\times\mathbb{T}^{1}$ in the
coordinates $(\phi_{1},\theta)$ or with $(-\pi,0)\times\mathbb{T}^{1}$
in the coordinates $(\phi_{2},\theta)$.  The transformation of
coordinates on $U_{1}\cap U_{2}$ is given in an obvious way, namely
\begin{eqnarray*}
  & \phi_{2}=\phi_{1}\phantom{\phantom{.-2\pi}} &
  \mbox{if\ }\,0<\phi_{1}<\pi,\ 0<\phi_{2}<\pi,\\
  & \phi_{2}=\phi_{1}-2\pi &
  \mbox{if\ }\,\pi<\phi_{1}<2\pi,\ -\pi<\phi_{2}<0.
\end{eqnarray*}

Any line bundle over a cylinder is equivalent to a trivial one. We
glue together $U_{1}\times\mathbb{C}$ and $U_{2}\times\mathbb{C}$ by a
transition function $\tau$ defined on $U_{1}\cap U_{2}$, with values
in $U(1)$. We put
\[
\tau=1\ \mbox{on\ }U_{12A},\ \tau
=e^{-iN\theta}\,\ \mbox{on\ }U_{12B},\ \mbox{where\ }
N\in\mathbb{Z}\ \mbox{is fixed}.
\]
Let $\mathfrak{V}_{N}$ denote the resulting line bundle over
$\mathbb{T}^{2}$.  Then a section in $\mathfrak{V}_{N}$ is determined
by a couple of complex functions $(\psi_{1},\psi_{2})$ defined on
$U_{1}$ and $U_{2}$, respectively, so that $\psi_{1}=\tau\psi_{2}$ on
$U_{1}\cap U_{2}$.  A connection in $\mathfrak{V}_{N}$ is determined
by a couple of one-forms $(\alpha_{1},\alpha_{2})$ defined
respectively on $U_{1}$ and $U_{2}$, with values in $i\mathbb{R}$, and
such that $\alpha_{2}=\alpha_{1}+\tau^{-1}\mbox{d}\tau$ on
$U_{1}\cap{}U_{2}$. Our choice is
\[
\alpha_{1}=\frac{iN}{2\pi}\,\phi_{1}\mbox{d}\theta,\ 
\alpha_{2}=\frac{iN}{2\pi}\,\phi_{2}\mbox{d}\theta.
\]
For the curvature (the two-form of the magnetic field) we get%
\hspace{1.2em}%
$\Omega=\mbox{d}\alpha=\linebreak%
iN/(2\pi)\,\mbox{d}\phi\wedge\mbox{d}\theta$,
and one has
\[
\frac{1}{2\pi i}\,\int_{\mathbb{T}^{2}}\Omega=N.
\]
Since $\tau$ takes its values in $U(1)$, the Hermitian structure in
$\mathfrak{V}_{N}$ is induced by the standard scalar product in
$\mathbb{C}$.

A crucial role in the formalism of Subsection~\ref{sec:triv_line_Mtilde}
is played by the family of functions $\Psi_{\gamma}$, $\gamma\in\Gamma$.
To find these functions we first need to describe the line bundle
$\tilde{\mathfrak{V}}_{N}=\pi^{\ast}\mathfrak{V}_{N}$. To this end,
let us cover $\mathbb{R}^{2}$ by two countable families of open strips,
namely $\tilde{U}_{1,k}=(2\pi k,2\pi(k+1))\times\mathbb{R}$ and $\tilde{U}_{2,k}=(\pi(2k-1),\pi(2k+1))\times\mathbb{R}$,
$k\in\mathbb{Z}$. Then $\pi(\tilde{U}_{1,k})=U_{1}$ and $\pi(\tilde{U}_{2,k})=U_{2}$,
$\forall k\in\mathbb{Z}$. A section of $\tilde{\mathfrak{V}}_{N}$
is determined by two countable families of functions, namely $\varphi_{1,k}$
defined on $\tilde{U}_{1,k}$ and $\varphi_{2,k}$ defined on $\tilde{U}_{2,k}$,
$k\in\mathbb{Z}$, such that\[
\varphi_{1,k}(x,y)=\varphi_{2,k}(x,y)\,\ \mbox{on\ }\big(2\pi k,\pi(2k+1)\big)\times\mathbb{R},\]
and\[
\varphi_{1,k}(x,y)=e^{-iNy}\varphi_{2,k+1}(x,y)\,\ \mbox{on\ }\big(\pi(2k+1),2\pi(k+1)\big)\times\mathbb{R}.\]
For a nowhere vanishing smooth section $\eta$ of $\tilde{\mathfrak{V}}_{N}$
we choose the families\[
\eta_{1,k}(x,y)=e^{ikNy},\ \eta_{2,k}(x,y)=e^{ikNy},\ k\in\mathbb{Z}.\]
Then $\eta$ determines a trivialization of $\tilde{\mathfrak{V}}_{N}$,
and one finds that\[
\Psi_{\gamma}(x,y)=e^{-imNy}\,\ \mbox{for\ }\gamma=(2\pi m,2\pi n)\in(2\pi\mathbb{Z})^{2}.\]

Below we consider the case with vanishing scalar potential so that
explicit computations are possible. Furthermore, we leave aside the
case $N=0$ which has particular properties but whose discussion is
very elementary. Thus the Hamiltonian $\tilde{H}$ describes a charged
particle in a homogeneous magnetic field on the plane,
\begin{equation}
  \tilde{H}=-\partial_{x}^{\,2}
  -\left(\partial_{y}+\frac{iN}{2\pi}\,x\right)^{\!2}\ \mbox{\ in\ }
  L^{2}(\mathbb{R}^{2},\mbox{d}x\mbox{d}y).
  \label{eq:Htilde_torus}
\end{equation}
The symmetry of $\tilde{H}$ is demonstrated by the fact that it
commutes with all magnetic translations $W_{\gamma}$,
$\gamma\in\Gamma$. One can establish the symmetry even under somewhat
more general circumstances, for any real $N$ (not necessarily an
integer). Defining the operators $W_{a,b}$, with $a,b\in\mathbb{R}$,
by the relation
\[
W_{a,b}\varphi(x,y)=e^{-iNay/(2\pi)}\varphi(x-a,y-b),
\]
one observes that $\tilde{H}W_{a,b}=W_{a,b}\tilde{H}$. Note that
$W_{a_{1},b_{1}}W_{a_{2},b_{2}}=e^{iNa_{2}b_{1}/(2\pi)}%
W_{a_{1}+a_{2},b_{1}+b_{2}}$ but if
$\gamma_{1},\gamma_{2}\in(2\pi\mathbb{Z})^{2}$ and $N\in\mathbb{Z}$,
then $W_{\gamma_{1}}W_{\gamma_{2}}=W_{\gamma_{1}+\gamma_{2}}$.

As is well known, using the Fourier transformation in the $y$ variable
one can decompose $\tilde{H}$ into a direct integral whose components
are unitarily equivalent to the Hamiltonian of the harmonic oscillator
with parameter values: $\hbar=1$, the mass equals $1/2$ and the
frequency $\omega=2|N|/(2\pi)$. Consequently, one can express the
Green function of $\tilde{H}$ in terms of the Green function
$\mathcal{G}_{z}^{\mathrm{ho}}$ of the harmonic oscillator (with the
indicated values of parameters).  One has
\begin{equation}
  \tilde{\mathcal{G}}_{z}(x_{1},y_{1};x_{2},y_{2})
  =\frac{1}{2\pi}\,\int_{\mathbb{R}}\,\mathcal{G}_{z}^{\mathrm{ho}}
  \negmedspace\left(x_{1}+\frac{2\pi k}{N},x_{2}
    +\frac{2\pi k}{N}\right)
  e^{ik(y_{1}-y_{2})}\,\mbox{d}k,
  \label{eq:Gtilde_int_Gho}
\end{equation}
and
\[
\mathcal{G}_{z}^{\mathrm{ho}}(x_{1},x_{2})
=\frac{1}{\sqrt{2\pi\omega}}\,
\Gamma\!\left(\frac{1}{2}-\frac{z}{\omega}\right)
D_{-\frac{1}{2}+\frac{z}{\omega}}\!\left(\sqrt{\omega}\, x_{>}\right)
D_{-\frac{1}{2}+\frac{z}{\omega}}\!\left(-\sqrt{\omega}\, x_{<}\right)
\]
where $D_{\nu}(x)$ is the parabolic cylinder function,
$x_{>}=\max\{x_{1},x_{2}\}$, $x_{<}=\min\{x_{1},x_{2}\}$ (see
\cite{groschesteiner} and references therein).

Let us write a one-dimensional representation $\Lambda$ of $\Gamma$ in
the form
\[
\Lambda(\gamma)=e^{2\pi i(\mu m+\nu n)}\,\ \mbox{for\ }\gamma
=(2\pi m,2\pi n)\in(2\pi\mathbb{Z})^{2}
\]
where $\mu,\nu\in[\,0,1)$ (here we change the meaning of the symbol
$\mu$). The Hamiltonian $H^{\Lambda}$ as a differential operator has
the same form as that given on the RHS of (\ref{eq:Htilde_torus}) but
it acts in the Hilbert space $\sH^{\Lambda}$ formed by measurable
functions $\varphi$ on $\mathbb{R}^{2}$ which are square integrable
over the fundamental domain $D=(0,2\pi)^{2}$ and satisfy, almost
everywhere on $\mathbb{R}^{2}$,
\[
\varphi(x+2\pi,y)=e^{2\pi i\mu}e^{-iNy}
\varphi(x,y),\ \varphi(x,y+2\pi)=e^{2\pi i\nu}\varphi(x,y).
\]

The Green function of $H^{\Lambda}$ can be derived with the aid of
formula (\ref{eq:formG_trivializable}). Using
(\ref{eq:Gtilde_int_Gho}) and the rapid decay of the parabolic
cylinder functions,

\[
D_{\nu}(x)=e^{-x^{2}/4}\, x^{\nu}\left(1+O\left(x^{-1}\right)\right)\ \mbox{as\ }x\to+\infty,\]
one can apply the Poisson summation in (\ref{eq:formG_trivializable})
to arrive at the equality\begin{eqnarray}
\mathcal{G}_{z}^{\Lambda}(x_{1},y_{1};x_{2},y_{2}) & = & \frac{1}{2\pi}\,\sum_{s=0}^{|N|-1}\sum_{k\in\mathbb{Z}^{2}}\,\mathcal{G}_{z}^{\mathrm{ho}}\!\!\left(\! x_{1}\!+\!2\pi k_{1}\!+\!\frac{2\pi(s+\nu)}{N},x_{2}\!+\!2\pi k_{2}\!+\!\frac{2\pi(s+\nu)}{N}\right)\nonumber \\
\noalign{\medskip} &  & \times\, e^{ik_{1}(Ny_{1}-2\pi\mu)-ik_{2}(Ny_{2}-2\pi\mu)+i(s+\nu)(y_{1}-y_{2})}\,.\label{eq:formG_torus}\end{eqnarray}
Note that\begin{equation}
\mathcal{G}_{z}^{\Lambda}(x_{1},y_{1};x_{2},y_{2})=e^{i\nu y_{1}}\mathcal{G}_{z}^{\Lambda=1}\!\left(x_{1}+\frac{2\pi\nu}{N},y_{1}-\frac{2\pi\mu}{N};x_{2}+\frac{2\pi\nu}{N},y_{2}-\frac{2\pi\mu}{N}\right)e^{-i\nu y_{2}}.\label{eq:GL_G1_torus}\end{equation}

A formula analogous to (\ref{eq:formG_torus}) also holds for the
kernel $p^{\Lambda}(t;x_{1},y_{1};x_{2},y_{2})$ of the Schr\"odinger
semigroup $\exp(-tH^{\Lambda})$, $t>0$. Let $p^{\mathrm{ho}}(t;x_{1},x_{2})$
denote the kernel of the Schr\"odinger semigroup $\exp(-tH^{\mathrm{ho}})$,
$t>0$, where $H^{\mathrm{ho}}$ is the Hamiltonian of the harmonic
oscillator (for the proper choice of parameters). Using the formula
for the Schr\"odinger semigroups (replacing $\mathcal{G}_{z}(\ldots)$
by $p(t;\ldots)$ in (\ref{eq:formG_torus})) one can compute the
traces. After some simple manipulations one finds that\begin{eqnarray*}
\Tr\!\left(e^{-tH^{\Lambda}}\right) & = & \int_{0}^{2\pi}\int_{0}^{2\pi}p^{\Lambda}(t;x,y;x,y)\,\mbox{d}x\mbox{d}y\\
 & = & |N|\int_{\mathbb{R}}p^{\mathrm{ho}}(t;u,u)\,\mbox{d}u\\
 & = & |N|\Tr\!\left(e^{-tH^{\mathrm{ho}}}\right)\\
 & = & \frac{2|N|}{\sinh(t\omega/2)}\,.\end{eqnarray*}
This equality makes it possible to compare the spectra of $H^{\Lambda}$
and $H^{\mathrm{ho}}$. Thus one deduces that

\begin{equation}
  \spec H^{\Lambda}
  =\left\{ \frac{|N|}{\pi}\!\left(\ell+\frac{1}{2}\right)\!;\,
    \ell=0,1,2,\ldots\right\} \label{eq:spec_HL_torus}
\end{equation}
where the multiplicity of each eigenvalue of $H^{\Lambda}$ equals
$|N|$. Of course, the spectrum of $H^{\Lambda}$ can be also derived by
solving directly the corresponding eigenvalue equation
\cite{cdeverdiere,onofri,sakamototanimura,alhashimiwiese}.

Equality (\ref{eq:spec_HL_torus}) particularly implies that all operator
$H^{\Lambda}$, $\Lambda\in\hat{\Gamma}$, are mutually unitarily
equivalent (this is not true for $N=0$). The corresponding unitary
mapping can be immediately deduced from (\ref{eq:GL_G1_torus}). Consider
the linear mapping
\[
T^{\Lambda}:\sH^{\Lambda=1}
\to\sH^{\Lambda},\ T^{\Lambda}\varphi(x,y)
=e^{i\nu y}\,\varphi\!\left(x+\frac{2\pi\nu}{N},
  y-\frac{2\pi\mu}{N}\right).
\]
One readily verifies that $T^{\Lambda}$ is well defined, unitary and
$H^{\Lambda}=T^{\Lambda}H^{\Lambda=1}(T^{\Lambda})^{-1}$.

\section*{Acknowledgments}

The authors wish to acknowledge gratefully partial support from the
following grants: Grant No.\ 201/09/0811 of the Czech Science
Foundation (P.\v{S}.) and Grant No.\ LC06002 of the Ministry of
Education of the Czech Republic (P.K.).

\end{document}